\documentclass[11pt]{amsart}

\usepackage{setspace} 
\usepackage{placeins}
\usepackage[subfigure]{graphfig}

\usepackage{amsthm}

\usepackage[all]{xy}
\usepackage{graphicx}
\usepackage{amsmath}
\usepackage{amssymb}
\usepackage{mathtools}
\usepackage{enumerate}
\usepackage{mathrsfs}
\usepackage{epsfig}
\usepackage{graphicx}
\usepackage{float}
\usepackage{epigraph}
\usepackage{bm}
\usepackage{bbm}

\usepackage{algpseudocode}
\usepackage{algorithmicx}
\usepackage{algorithm}
\usepackage{epstopdf} %JLK
\usepackage{natbib}

\usepackage{color,soul}

\numberwithin{equation}{section}

\newtheorem{thm}{Theorem}[section]

\newtheorem{lem}[thm]{Lemma}
\newtheorem{prop}[thm]{Proposition}
\newtheorem{rem}{Remark}[section]

\newcommand{\eq}[1]{(\ref{#1})}

\renewcommand{\Re}{\operatorname{\rm Re}}
\renewcommand{\Im}{\operatorname{\rm Im}}

\newcommand{\beqast}{\begin{eqnarray*}}
\newcommand{\eqast}{\end{eqnarray*}}
\newcommand{\beqa}{\begin{eqnarray}}
\newcommand{\eqa}{\end{eqnarray}}

\newcommand{\bbe}{\begin{equation}}
\newcommand{\ee}{\end{equation}}

\renewcommand{\Re}{\operatorname{\rm Re}}
\renewcommand{\Im}{\operatorname{\rm Im}}

\newcommand{\bE}{{\mathbb E}}
\newcommand{\bN}{{\mathbb N}}

\newcommand{\bR}{{\mathbb R}}
\newcommand{\bC}{{\mathbb C}}
\newcommand{\bZ}{{\mathbb Z}}

\newcommand{\cC}{{\mathcal C}}
\newcommand{\cL}{{\mathcal L}}

\newcommand{\cV}{{\mathcal V}}

\newcommand{\al}{\alpha}
\newcommand{\be}{\beta}
\newcommand{\De}{\Delta}
\newcommand{\de}{\delta}
\newcommand{\eps}{\epsilon}

\newcommand{\lp}{\lambda_+}
\newcommand{\lm}{\lambda_-}
\newcommand{\La}{\Lambda}
\newcommand{\mum}{\mu_-}
\newcommand{\mup}{\mu_+}
\newcommand{\sg}{\sigma}

\newcommand{\om}{\omega}
\newcommand{\ze}{\zeta}

\newcommand{\ga}{\gamma}
\newcommand{\gap}{\gamma_+}
\newcommand{\gam}{\gamma_-}
\newcommand{\Ga}{\Gamma}

\newcommand{\omp}{\omega_+}
\newcommand{\omm}{\omega_-}

\newcommand{\barDe}{\bar\Delta}
\newcommand{\htphi}{\widehat{\widetilde{\varphi}}}
\newcommand{\dd}{\partial}

\newcommand{\hG}{{\hat G}}
\newcommand{\hV}{{\hat V}}

\newcommand{\bfo}{{\bf 1}}

\newcommand{\norm}[1]{\lVert#1\rVert}
\newcommand{\cspan}[1]{\overline{\text{span}}\{#1\}}

\begin{document}
\begin{titlepage}
\title[SINH-acceleration for B-spline option pricing]{SINH-acceleration for B-spline projection with Option Pricing Applications }
%
%\author[
%Svetlana Boyarchenko and
%Sergei Levendorski\u{i}]
%{
%Svetlana Boyarchenko and
%Sergei Levendorski\u{i}}

\author[
 Boyarchenko et al.]
{
Svetlana Boyarchenko,  Sergei Levendorski\u{i}, J. Lars Kirkby, and Zhenyu Cui}

\thanks{
\emph{S.B.:} Department of Economics, The
University of Texas at Austin, 2225 Speedway Stop C3100, Austin,
TX 78712--0301, {\tt sboyarch@eco.utexas.edu} \\
\emph{S.L.:}
Calico Science Consulting. Austin, TX,
  {\tt
levendorskii@gmail.com}\\
\emph{J.K.:} ISYE, Georgia Institute of Technology, 755 Ferst Dr., Atlanta, GA 30313,
{\tt jkirkby3@gatech.edu} 
\\
\emph{Z.C..:} School of Business, Stevens Institute of Technology, Babbio Dr, Hoboken, NJ 07030, {\tt zcui6@stevens.edu}
}

\maketitle
\thispagestyle{empty}

\begin{abstract}
We clarify the relations among different Fourier-based approaches to option pricing,
and improve the B-spline probability density projection method using the sinh-acceleration  technique.
This allows us to efficiently separate the control of different sources of errors
better than the FFT-based realization allows; in many cases,
the CPU time decreases as well.  We demonstrate the improvement of the B-spline projection method through several numerical experiments in option pricing, including European and barrier options, where the SINH acceleration technique proves to be robust and accurate.

\end{abstract}

\tableofcontents
\end{titlepage}

%%%%%%%%%%%%%%%%%%%%%%%%%%%%%%%%%
\section{Introduction}
%The aim of the paper is to clarify the relations among different Fourier-based approaches to option pricing,
%and improve the B-spline density projection technique using the sinh-acceleration technique. 

\subsection{Fourier transform method in finance}\label{ss:FT in Finance} For two centuries, 
the Fourier transform is one of the most powerful tools in various branches of mathematics, natural sciences and engineering. The Fourier transform is a basis for very deep theoretical results in analysis, e.g., in the study 
of   differential and pseudo-differential operators (pdo)  and boundary problems for pdo's\footnote{See., e.g., \cite{eskin,Hormander,Duistermaat95,DegEllEq}
%,LevPan}  
and the bibliographies theirein.}, 
and in probability, where  it appears in the form of characteristic functions of random variables. In many cases,
application of the Fourier transform technique results in explicit analytical formulas, in the form of integrals, multi-dimensional ones included. The host of efficient numerical methods for evaluation of these integrals has been developed in the numerical analysis literature a long time ago\footnote{See, e,g., \cite{AbWh,AbWh92OR,stenger-book,AbateValko04,stengerreview00} and the bibliographies therein.}.  In Finance, the prices of  derivative securities are  defined as either solutions of boundary problems for differential or integro-differential equations or as expectations of stochastic payoffs and streams of payoffs under a risk-neutral measure. The former definition  is used in seminal papers by Black, Scholes and Merton; being based on hedging arguments, this definition is theoretically sound in the case of diffusion models only.  The definition of the price as the expectation of the stream of payoffs under
an equivalent martingale measure chosen for pricing originated in works of Ross, Harrison-Kreps and Harrison-Pliska later. This definition  can be applied to many jump-diffusion models. Under certain regularity conditions, one can derive the boundary problem for the Kolmogorov backward integro-differential equation (a special case of pseudo-differential equation) for the price, and find the price that solves the  problem. Thus, 
essentially any pricing problem can be regarded as a straightforward exercise in applying  the Fourier transform technique using available mathematical tools - both theoretical and numerical, from analysis and probability. For complicated derivative products, some small new non-trivial twists might be necessary but,  for the most part,
the only difficulty stems from the fact that finance requires very fast calculations, and therefore the numerical methods available in the literature may need some adjustments. 
So, 
it is rather surprising that the Fourier transform methods had to wait for more than two decades to be applied to finance. 

The first application is due to  \cite{Heston93}, who derived the formula for prices of European options in the stochastic volatility model that bears his name since. Although the underlying process in the Heston model is a 2D diffusion process, the prices of European options in the model {\em conditioned on the unobserved volatility factor} can be interpreted as the prices of
European options in a 1D L\'evy model with the characteristic exponent of a rather complicated structure. Hence,
it is fair to say that \cite{Heston93} is the first paper where the Fourier transform method is applied to pricing European options in L\'evy models. Formally, the first applications of the Fourier transform method to pricing European options with L\'evy models are \cite{BL-FT,carr-madan-FFT}, and the real differences among \cite{Heston93,BL-FT,carr-madan-FFT}
are in the realizations of the Fourier transform. The analytical expression derived in \cite{Heston93} is natural from
the viewpoint of probability but very inconvenient for an efficient numerical realization (a representation of the same kind was used in \cite{DPS} to price options in affine jump-diffusion models). The analytical expressions derived in
\cite{BL-FT,carr-madan-FFT} later are natural ones from the point of view of analysis, and, with numerous modifications and extensions, are used in numerous applications to finance and insurance ever since. From the point of view of applications, the
main differences among methods used in various papers are in the quality of error control and speed. 
In each case, for the development of an efficient numerical procedure, it is necessary to derive
sufficiently accurate error bounds and give recommendations for the choice of parameters of the numerical scheme
which allow one to control all sources of errors.  It is well-known that in the case of stable L\'evy processes/distributions
 accurate and fast calculations are extremely difficult, especially if the Blumenthal-Getoor index $\al$ is close to 0 (or to 1,
 in the asymmetric case)
 when essentially all known methods fail 
  (see \cite{ConfAccelerationStable} for examples and the bibliography). Hence, 
if the underlying L\'evy process (or conditional distribution) is close
to a stable one in the sense that the rate of tails of the L\'evy density decay slowly and/or the asymptotics of 
the L\'evy density near 0 are  the same as the one of the stable L\'evy process of index $\al$ close to 0 (or to 1,
in the asymmetric case). A recently developed methodology based on appropriate conformal deformations
of the contours of integration, corresponding changes of variables and application of the simplified trapezoid rule
allows one to develop efficient numerical methods. 

\subsection{Structure of the paper and outline of the main ideas}\label{ss:structure}
In the main body of the paper, we show how to use the conformal deformation technique to improve the performance and robustness
of the B-spline density projection method (PROJ). Then we consider applications of PROJ to barrier options with discrete monitoring.
Since each step of backward induction and calculation of  the coefficients in PROJ method can be naturally interpreted as pricing of a European option, it is natural to consider 
first the pricing of European options and the evaluation of probability distributions in L\'evy models  (Sect. \ref{ss:Euro}-\ref{ss:conf}), and then
applications and modifications of pricing  European options in backward induction procedures (Sect. \ref{ss:barr_back}). 
In Sect. \ref{ss: Bsplines}, we recall the main ingredients of the B-spline density projection method (PROJ method) developed in \cite{Ki14,KiDe14, Ki14B, Ki16A} to price options of several classes, and explain the advantages 
of PROJ as compared to the COS method. The relation of the approximation procedures   to the spectral filtering and implications for the accuracy pricing in applications to risk management are discussed in 
Sect. \ref{appr_spectral_filter}.

The numerical scheme in \cite{Ki14} and its extensions rely on the Fast Fourier transform (FFT) method
as in hundreds (possibly, thousands) of other papers in quantitative finance. However, as it is explained in detail in \cite{single,BLdouble}, the standard application of FFT does not allow one to control all sources of errors at a low CPU and memory cost. In Sect. \ref{ss:conf}, we recall the drawbacks of FFT and the remedy suggested in \cite{single},
namely, the choice of grids in the state and dual spaces, of different sizes and steps, and application of FFT to several
subgrids so that the error tolerance can be satisfied at approximately minimal CPU and memory costs.
The same trick can be applied to improve the performance of PROJ, but in this paper, we use a different improvement
used in \cite{DeLe14,LeXi12,AsianGammaSIAMFM} to price barrier options with discrete monitoring and Asian options 
with discrete sampling. The idea is to completely separate the control of errors of the approximation of the transition operator at each time step as a convolution operator whose action is
realized using the fast discrete convolution algorithm, and errors of the calculation of the elements of the discretized pricing kernel. The latter elements are calculated one-by-one (the procedure admits an evident parallelization) without resorting to FFT technique. Instead, appropriate conformal deformations of the contour of integration in the formula for the elements, with the subsequent change of variables, are used to greatly increase the rate of convergence. The changes of variables are such that the new integrand is analytic in a strip around the line of integration, hence, the integral can be calculated with accuracy E-14 to E-15 using the simplified trapezoid rule with a moderate or even small number of points.
%The  details are Sect. \ref{ss:conf}.

Contour deformations  are among standard tools in numerical analysis: see, e.g.,
\cite{stenger-book,stengerreview00}. In particular, there is a universal albeit not quite straightforward saddle point method
- see, e.g., \cite{Fedoryuk} for the general background. As numerical examples in \cite{IAC} demonstrate, an approximate
realization of the saddle point method applied in \cite{Carr-Madan-saddlepoint} to price deep OTM European options produces errors of the order of 0.005.
In the same situations, the fairly simple  fractional-parabolic deformation method developed in \cite{iFT0,iFT} satisfies the error tolerance of the order of E-15 at very small CPU cost, and a more efficient sinh-acceleration method suggested in \cite{BoyarLevenSinh19}
requires microseconds in MATLAB to achieve this level of accuracy. In some cases, the composition of the fractional-parabolic and sinh-deformations may decrease the complexity of the method: see \cite{Sinh} for applications to  evaluation of special functions. The families of deformations suggested in \cite{iFT0,iFT,BoyarLevenSinh19} are flexible and easy to implement; both were used to price derivative securities of various classes
and calculate probability distributions (variations of these families which allow one to evaluate stable distribution are suggested and used in \cite{ConfAccelerationStable})
and can be used in many other situations as well. We explain the details and give further references in 
Sect. \ref{ss:conf}. 

In Sect. \ref{s:SINH-PROJ}, we derive the formulas for the coefficients using  appropriate conformal deformations and changes of variables. As in \cite{paired,Contrarian}, we cross some poles of the integrand and apply the residue theorem.
In application to pricing European options, the novelty of the situation in the present paper is that 
\begin{enumerate}[(1)]
\item there are infinitely many poles which need to be crossed whereas in \cite{paired,Contrarian}, the number of poles in the domain of analyticity is finite; 
\item
the result is a sum of an integral and an infinite sum of explicitly calculated terms\footnote{Calculation of the Wiener-Hopf factors for rational or meromorphic symbols, hence, in L\'evy models and with rational or meromorphic characteristic exponents is easily reducible to finite and infinite sums in this way because the remaining integral vanishes as
the contour of integration moves to infinity - 
see \cite{NG-MBS,amer-put-levy,
KuznetsovWien10,one-sidedCDS}; also, the reduction to an infinite sum is possible if barrier options with discrete monitoring in the Black-Scholes model are priced \cite{FusaiBarr}};  
\item
for calculations to be efficient, it may be non-optimal to cross all the poles;
\item
 truncation must be done not of the sum in infinite trapezoid rule only but of the infinite sum of residues as well.
 \end{enumerate}

In Sect. 
\ref{s:barrier pricing}, we apply the B-spline projection method to pricing of several types of barrier options, and
discuss the differences with other methods that realize the main block of  backward induction in the state space.
For the comparison with the alternative approach based on the calculations in the dual space  \cite{FeLi08,FusaiGermanoMarazzina},
see \cite{DeLe14}. 
We derive error bounds and give recommendation for the choice of parameters
of the numerical schemes. Numerical examples are in Sect.  \ref{s:numer}.   In Sect. \ref{s:concl}, we summarize the results of the paper and outline other possible applications of the method of the paper.

\section{Pricing European options and barrier options with discrete monitoring}\label{s:Euro_barrier}
\subsection{Pricing European options in L\'evy models and evaluation of probability distributions}\label{ss:Euro}
 The analytical expressions in 
\cite{BL-FT,carr-madan-FFT} are derived using the straightforward approach used in analysis; the difference is that in 
\cite{BL-FT}, the Fourier transform w.r.t. the log-spot is used, and in \cite{carr-madan-FFT}, w.r.t. the log-strike.
Let the riskless rate $r$ be constant, let $X$ be a L\'evy process on $\bR$ under the risk-neutral measure chosen for pricing, and let
$\bE[e^{iX_t\xi}]=e^{-t\psi(\xi)}$ be
the characteristic function of $X_t$ (note that we use the definition of the characteristic exponent $\psi$ as in \cite{BL-FT,KoBoL,NG-MBS},
which is marginally different from the definition in \cite{Sa99}). 
Let  
\bbe\label{hG}
\hG(\xi)=\int_\bR e^{-ix\xi}G(x)d\xi
%(\cF G)(\xi)=
\ee
 be the Fourier transform of the payoff function (this is the definition of the Fourier transform used
 in analysis and physics; the definition used in probability differs in sign: $e^{ix\xi}$ instead of $e^{-ix\xi}$). Expanding $G$ in the Fourier integral
 \bbe\label{iFT}
 G(x)=\int_{\Im\xi=\om} e^{ix\xi}\hG(\xi)d\xi,
 \ee
 where $\om$ is chosen so that $\hG$ and $\psi$ are  well-defined on $\{\Im\xi=\om\}$,    and applying the Fubini theorem, 
  one easily derives the formula for the option price
 \bbe\label{expect}
 V(T;x)=e^{-rT}\bE^x[G(X_T)]=(2\pi)^{-1}\int_{\Im\xi=\om}e^{ix\xi-T(r+\psi(\xi))} \hG(\xi)d\xi.
 \ee
 In order that the integral absolutely converges, the product 
 $e^{-T\psi(\xi)}\hG(\xi)$ must decay sufficiently fast as $\Re\xi\to\pm\infty$, and the application of the Fubini's theorem is justified. Alternatively, integration by parts can be used to regularize the integral. See \cite{NG-MBS,iFT} for details.
 
 The probability distribution $p_t(x)$ of $X_t$ can be interpreted as the price of the option with the payoff $\de_0$, in
 the model with the dual process $-X_t$ and 0 interest rate. Since the Fourier transform of the delta function is 1,
 \bbe\label{pt}
 p_t(x)=(2\pi)^{-1}\int_{\Im\xi=-\om}e^{ix\xi-t\psi(-\xi)}d\xi=(2\pi)^{-1}\int_{\Im\xi=\om}e^{-ix\xi-t\psi(\xi)}d\xi,
 \ee
 where the line $\{\Im\xi=\om\}$ must be in the strip of analyticity of $\psi$ around the real axis.
  The reader observes that, in each case, the derivation is at the level of a simple exercise in complex calculus once the formula for
 the characteristic function is known. For applications to finance, efficient (meaning: sufficeintly 
 accurate and fast) numerical procedures for the evaluation of integrals are needed. In many cases of interest, the two requirements are
 difficult to reconcile. A good and important example is the popular Variance Gamma model (VG model) introduced to finance in \cite{MaSe90,MGC98}.
 In the VG model, the characteristic exponent admits the asymptotics 
\bbe\label{asympVG}
\psi(\xi)\sim c_\infty \ln|\xi|-i\mu\xi,
\ee
as $\xi\to \infty$ remaining in the strip of analyticity, where $c_\infty>0$ and $\mu\in\bR$ is the drift.  
Letting $x'=-x+t\mu$, we see that the integrand in \eq{pt} has the asymptotics 
 \bbe\label{polf}
 f(\xi)\sim e^{ix'\xi} |\xi|^{-tc_\infty},
 \ee
 as $\xi\to\infty$ in the strip of analyticity of the characteristic exponent. Hence, if $x'=0$ and $tc_\infty\le 1$,  the integral diverges; the pdf is unbounded as $x'\to 0$. If $tc_\infty\in (1,2)$, $p_t$ has a cusp at $x=\mu t$, and if 
 $tc_\infty=2$, then a kink. If $x'\neq 0$, then, for any $t>0$, one can integrate by parts and reduce the initial integral to an absolutely convergent integral. 
For the parameters of the VG model documented in the empirical literature and $t$ of the order of 0.004 (one day) and even larger, $tc_\infty$ is very small, and, therefore, the truncation of the integral or the infinite sum in the infinite trapezoid rule
at a moderate level implicit in the recommendations given in   \cite{carr-madan-FFT} (CM method)
produces large errors. 
Using more accurate prescriptions for the parameter choice, one can satisfy a small error tolerance but at extremely large CPU and memory costs. See examples in \cite{iFT0,iFT,BoyarLevenSinh19}. In addition, the CM method introduces 
 an additional source of errors because the interpolation across the strikes of vanilla options was used.
 The interpolation is necessary so that  FFT can be applied. It was widely believed that the possibility to apply FFT
 was the great advantage of CM method because it allowed for a simultaneous calculation of vanilla prices at many strikes. However, in practice, the number of vanilla options of the same maturity is not large, and an
 accurate choice of the parameters of a good numerical integration procedure allows one to do calculations faster than
 CM method allows. Apart from an unnecessary interpolation error which is introduced by CM method, the ad-hoc rigid prescription for the parameter choices in CM method as it is understood in the literature leads to serious systematic errors,
 in applications to calibration and risk management in particular. See \cite{pitfalls,HestonCalibMarcoMe,HestonCalibMarcoMeRisk} for examples. In particular, in \cite{HestonCalibMarcoMe,HestonCalibMarcoMeRisk}, it is demonstrated that the model with the  parameters calibrated to the volatility smiles reproduces the smile poorly; even the smile accurately calculated for the calibrated parameters cannot be reproduced if the same method is used for the calibration. The same effect albeit to a lesser degree is documented for a popular COS method \cite{FaOo08}. In COS method,   an additional source of errors is introduced and the error control is inefficient. See \cite{iFT0,iFT,pitfalls} for the theoretical analysis and examples. 
 As explained in \cite{HestonCalibMarcoMe,one-sidedCDS}, an inaccurate pricing procedure used for calibration purposes
 will not recognize the correct parameter set ({\em sundial calibration}) and will find a local minimum where the true
 calibration error and error of the method almost cancel one another: the calibration procedure will see {\em ghosts} at
 the boundary of the region in the parameter space where the numerical method performs reasonably well ({\em ghost calibration}).
 If a method is rather inaccurate, this region is small, and, therefore, the calibration ``succeeds" to find a presumably 
 good set of the parameters of the model very fast. Hence, the procedure is ``very efficient".

  The papers \cite{Heston93,BL-FT} utilized reasonably accurate but rather inefficient numerical methods: in the first paper,
 the analytical representation was inconvenient for the development of efficient numerical procedures; in 
the second paper, it was suggested to  use   standard real-analytical numerical integration procedures such as the trapezoid and Simpson rules. Both methods guarantee only polynomial rate of convergence of the numerical scheme.
However, as it was remarked in \cite{single}, both rules can be interpreted as weighted sums of the simplified trapezoid
rule, hence, in fact, both have exponential rates of convergence; the rate of convergence is much worse than
the one of the simplified trapezoid rule. 
  Thus, in effect, the real-analytical recommendation for the choice of the step $\ze$ in
 \cite{BL-FT} requires unnecessarily small $\ze$, hence, an unnecessarily large number of terms, to satisfy the given error tolerance. 
Efficient realizations of the pricing formula in exponential L\'evy models (as well as in the Heston model and many other affine and quadratic models) are based on the wonderful property of the infinite trapezoid rule, namely, the exponential rate of decay of the discretization error as a function of $1/\ze$, where $\ze$ is the step, and  the efficient
integral  bound for the error of the infinite trapezoid rule. For $\mum<\mup$, denote 
$S_{(\mum\mup)}:=\{\xi\ | \Im\xi\in (-d,d)\}$. Let $f$ be analytic in a strip
$S_{(-d,d)}$, where $d>0$,  and decay at infinity sufficiently fast so that
%\bbe\label{condAHardy}
$\lim_{A\to \pm\infty}\int_{-d}^d |f(i a+A)|da=0,$
%\ee
and the Hardy norm
\bbe\label{Hnorm}
H(f,d):=\|f\|_{H^1(S_{(-d,d)})}:=\lim_{a\downarrow -d}\int_\bR|f(i a+ y)|dy+\lim_{a\uparrow d}\int_\bR|f(i a+y)|dy
\ee
is finite (we call $H(f,d)$ the Hardy norm following \cite{FeLi08}; the standard definition of the Hardy norm is marginally different). We write $f\in H^1(S_{(-d,d)})$. The integral
\bbe\label{INT}
I=\int_\bR f(\xi)d\xi
\ee
can be evaluated using the infinite trapezoid rule: for any $\ze>0$,  
\bbe\label{inftrap}
I\approx \ze\sum_{j\in \bZ} f(j\ze).
\ee
\begin{lem}[\cite{St93}, Thm.3.2.1]\label{lem:inf_trap}
The error of the infinite trapezoid rule \eqref{inftrap} admits an upper bound 
\bbe\label{Err_inf_trap}
{\rm Err}_{\rm disc}\le H(f,d)\frac{\exp[-2\pi d/\ze]}{1-\exp[-2\pi d/\ze]}.
\ee
\end{lem}
Lemma \ref{lem:inf_trap} was introduced to finance in \cite{Lee04} to price European options, and later, in 
\cite{FeLi08}, to price barrier options using the Hilbert transform. Note that in applications, the infinite sum
must be truncated, and the simplified trapezoid rule applied.  
In \cite{iFT}, an  analysis of the discretization and truncation errors for several classes of L\'evy models
is used to give accurate general recommendations for the choice of the parameters of the numerical scheme
given the error tolerance.

 The exponential decay of the discretization error of the infinite trapezoid rule makes the control of the
 discretization error fairly simple if the strip of analyticity is not too narrow. In application to pricing in L\'evy models, an equivalent condition is: the rate of the exponential decay of the  tails of the L\'evy density at infinity is not small.
  However, if $f(y)$ decays slowly as $y\to\pm \infty$, the number of terms
 in the simplified trapezoid rule
 \bbe\label{inftrapsimp}
I\approx \ze\sum_{|j|\le N} f(j\ze)
\ee
necessary to satisfy even a moderate error tolerance $\eps$ can be extremely large. See \cite{iFT0,iFT,BoyarLevenSinh19} for examples. The simplest example is the evaluation of the probability distribution function in the VG model. If $x'=0$ and $tc_\infty\le 1$, the integral diverges.
If $x'\neq 0$, then, for any $t>0$, one can reduce to absolutely convergent integral integrating by parts; the better way is to use the summation by parts in \eq{inftrap} as in
 \cite{iFT0,Contrarian} to reduce to the absolutely converging series and significantly decrease the number of terms necessary to satisfy the given error tolerance. There exist numerous similar acceleration schemes (Euler acceleration and other - see, e.g., \cite{AbWh,AbWh92OR,AbateValko04} and the bibliographies therein). Note that
 the errors of the acceleration schemes are not easily controlled, and one expects that these schemes are reliable only if
 the derivatives of the non-exponential factor in the integrand decay faster than the integrand itself. 
 This is, apparently, not the case if $\psi(\xi)$ increases as $|\xi|$ and faster as $x\to\infty$ along the line of integration.
 Thus,  in models of infinite variation, different tools should be applied; the conformal acceleration method
 explained below is superior to other acceleration methods.
 
 \subsection{Sinh-acceleration}\label{ss:conf}
 Fortunately, as it is noticed in \cite{iFT0,iFT,BoyarLevenSinh19}, all models popular in finance (bar the VG model) enjoy the following key properties. The characteristic exponent can be represented in the form
 \bbe\label{psi0}
 \psi(\xi)=-i\mu\xi+\psi^0(\xi),
 \ee
 where $\psi^0$ 
 is analytic in a union of a strip $S_{(\mum,\mup)}$ and  cone $\cC$ around the real axis s.t.  \bbe\label{asymppsi}
 \psi^0(\xi)\sim c_\infty({\mathrm{arg}}\,\xi)|\xi|^\nu,
 \ee
as $\xi\to\infty$ remaining in the cone,
 where $\nu\in (0,2]$.\footnote{In \cite{KoBoL,NG-MBS}, it is required that \eq{asymppsi} holds as $\xi\to \infty$ remaining in
 the strip. The corresponding more general class of L\'evy processes is called Regular L\'evy processes of
 exponential type - RLPE.} Furthermore, there exists a sub-cone $\bR\subset \cC'\subset \cC$ such that
 \bbe\label{asymppsipos}
 \Re c_\infty({\mathrm{arg}}\,\xi)>0,\quad \xi\in \cC'\setminus\{0\}.
 \ee
 In \cite{BoyarLevenSinh19}, the real axis can be at the boundary of the cone, the conditions \eq{asymppsi}-\eq{asymppsipos} are relaxed but additional conditions are imposed; the conditions  omitted here are needed for the application of the same technique
 to the evaluation of the Wiener-Hopf factors and efficient pricing of barrier options and lookbacks, with continuous monitoring. In the case of the VG model, $\psi^0$ increases as $\ln|\xi|$.
 
 For $-\pi/2\le \gam<0<\gap\le \pi/2$, define
 $\cC_{\gam,\gap}:=\{\xi\in\bC\ |\ {\mathrm{arg}}\,\xi\in (\gam,\gap)\cup (\pi-\gam, \pi-\gap)\}$.
 For basic types of models used in finance, conditions \eq{asymppsi}-
 \eq{asymppsi} hold with $\cC=\cC_{-\pi/2,\pi/2}$, and 
\eq{asymppsipos} holds with $\cC'=\cC_{-\ga_\nu,\ga_\nu}$, where $\ga_\nu=\min\{1,1/\nu\}\pi/2$.
Asymmetric cones naturally arise in the completely asymmetric
case of the processes of the generalized
 Koponen family  constructed in  \cite{genBS,KoBoL} (and called later KoBoL in \cite{NG-MBS}); the cones are symmetric for a subclass of KoBoL given the name CGMY in \cite{CGMY}. We keep the initial labels for the parameters of KoBoL: for $\nu\in (0,2), \nu\neq 1$, $c>0$, $\lm<0<\lp$,
 \begin{equation}\label{kbl}
\psi(\xi)=-i\mu\xi+c\Gamma(-\nu)[\lp^\nu-(\lp+i\xi)^\nu+(-\lm)^\nu-(-\lm-i\xi)^\nu],
\end{equation}
 therefore, for $\xi$ in the right half-plane, $c_\infty({\mathrm{arg}}\,\xi)=-c\Gamma(-\nu)\exp[i\nu{\mathrm{arg}}\,\xi]$.
 Hence, \eq{asymppsipos} holds with $\ga_\nu=\min\{1,1/\nu\}\pi/2$. In the case $\nu=1$, the formulas are different
 \cite{genBS,KoBoL,NG-MBS}.
 
 As another prominent example, the characteristic exponents of Normal Tempered Stable L\'evy processes  constructed in \cite{B-N-L} are given by
 \begin{equation}\label{NTS2}
\psi(\xi)=-i\mu\xi+\de[(\al^2+(\xi+i\be)^2)^{\nu/2}-(\al^2-\be^2)^{\nu/2}],
\end{equation}
where $\nu\in (0,2)$, $\de>0$, $|\be|<\al$; for $\nu=1$, this is the characteristic exponent of the Normal Inverse Gaussian process (NIG) \cite{BN98}.  It is evident that \eq{asymppsipos} holds with $\ga_\nu=\pi/2$. 
  
As it was demonstrated in  \cite{iFT0,iFT,BoyarLevenSinh19}, property \eq{asymppsipos} allows one to greatly increase the rate
of decay of integrands at infinity, hence, the rate of convergence of the infinite sums in the infinite trapezoid rule, 
 using appropriate conformal deformations of the line of integration
 and the corresponding changes of variables.
 We denote the conformal maps  $\chi$ and $\chi^\pm$, and the contours
$\cL=\chi(\bR)$ or $\cL^\pm=\chi^\pm(\bR)$. We use the subscript $+$ and $-$ if we wish to stress the fact that the wings of the contour point upward and 
downward, respectively. If needed, the parameters of a deformation and the corresponding change
of variables are included as subscripts. The integrand in the new variable is denoted $g(y)$. After the change of the variables, we apply the simplified trapezoid rule. In \cite{iFT0,iFT,BoyarLevenSinh19}, three families of deformations were suggested and used. The families are given by relatively simple formulas, which allows one not only to derive efficient error bounds and give accurate recommendations for the choice of the parameters of the numerical schemes for evaluation of 1D integrals
but perform similar deformations for integrals in dimensions 2, 3 and even 4. See \cite{paraLaplace,Contrarian} for applications to pricing barrier and lookback options, and hedging in L\'evy models. Other applications of the same
technique to quantitative finance can be found in \cite{paraHeston,DeLe14,LeXi12,paired,one-sidedCDS,AsianGammaSIAMFM,BarrStIR}; in \cite{ConfAccelerationStable}, the reader can find
variations of the same families suitable for efficient evaluation of stable distributions and
general recommendations for  approximately optimal choices of the family.  If the cone $\cC_{\ga_\nu}$ is not too ``narrow", equivalently, $\ga_\nu$ is not too small, then the sinh-acceleration 
  \bbe\label{eq:sinh}
\chi(y)=i\om_1+b\sinh(i\om+y), 
\ee
where $\om_1\in \bR, \om\in (-\pi/2,\pi/2)$ and $b>0$, is the best choice. 
If $\om>0$, the wings of $\cL$ point upward, and if $\om<0$, the wings point downward. 
Set $x'=x+\mu T$ and $z(y)=\sinh(i\om+y)$, then, after the contour  deformation and change of variables, we
reduce \eq{expect} to \bbe\label{expect00}
 V(T;x)=e^{-rT}\bE^x[G(X_T)]=(2\pi)^{-1}\int_{\bR}e^{ix'\xi-T(r+\psi^0(\xi))} \hG(\xi)d\xi.
 \ee
The parameters are chosen so that
\begin{enumerate}[(1)]
\item
the contour $\chi(\bR)$ is a subset of the domain of analyticity, and $\psi^0(\chi(y))$
decays fast as $y\to\pm\infty$, hence, $\om\in (-\ga_\nu,\ga_\nu)$;
\item
 the oscillating factor becomes fast decaying one, hence, $\om\in (0,\ga_\nu)$ if $x'>0$, and
 $\om\in (-\ga_\nu,0)$ if $x'<0$. If $x'=0$, any $\om\in (-\ga_\nu,\ga_\nu)$ is admissible, and, typically, $\om=0$ is the
 best choice.
\end{enumerate}
The largest gain in the CPU and memory costs is in the case of the VG model, at $x'\neq 0$. 
Let $\eps>0$ be the error tolerance and set $E=\ln(1/\eps)$. Instead of \eq{asymppsi}, we have
\bbe\label{asymppsiVG}
 \psi^0(\xi)\sim c_\infty({\mathrm{arg}}\,\xi)\ln|\xi|,
 \ee
 hence, the complexity of the scheme based on
 the simplified trapezoid rule without acceleration is of the order of $O(E \eps^{1/(1-Tc_\infty)})$, where $\eps>0$ is the error tolerance and $E=\ln(1/\eps)$;  if the summation by parts is used $n$ times, then of the order of $O(E \eps^{1/(n-Tc_\infty)})$,
 and if the sinh-acceleration is  used, then of the order of $O(E\ln(E/|x'|)$. If $x'=0$,  the complexity of the scheme
is of the order of  $O(E^2)$. In KoBoL and NIG models, the complexity reduces from $O(E^{1+a})$, where $a>0$,
 to $O(E\ln E)$, for any $x'$.

 \subsection{Pricing barrier options with discrete monitoring in L\'evy models, by backward induction}\label{ss:barr_back}
 
 As an example, consider double barrier options with the maturity date $T$ and
 barriers $H_\pm=e^{h_\pm}, h_-<h_+$. If, at any of the monitoring dates $t_0(=0)<t_1<t_2<\cdots <t_{N-1}<t_N:=T$, the price
 $S_{t_j}=e^{X_{t_j}}$ of the underlying is outside the interval $(H_-, H_+)$, the option expires worthless. If neither of the barriers is breached, at maturity date $T$, the option payoff is $G(X_T)$. Let $U=(h_-, h_+)$, and note that letting $h_-=-\infty$ (resp., $h_+=+\infty$),
 we obtain barrier options with one barrier. The riskless rate $r\in \mathbb R$ is constant.  
 
 Denote by $V_n(x)$ the option price at time $t_n$ and $X_{t_n}=x$, by $P_t=e^{-t\psi(D)}$ the transition operator, and set $\barDe_n=t_{n+1}-t_n$, $n=0,1,\ldots, N-1$.
 Typically, the monitoring dates are equally spaced, hence, $\barDe=\barDe_n$ is independent of $n$.
 Writing down the price of the option as an expectation, and applying the law of iterated expectations, one obtains the following straightforward scheme:
 \begin{enumerate}[(1)]
 \item
 set $V_N=\bfo_U G$;
 \item
 in the cycle $n=N-1,N-2,\ldots,0$, for $x\in U$, calculate
 \bbe\label{indstep}
 V_{n}(x)=e^{-\barDe_n}(P_{\barDe_n}V_{n+1})(x);
 \ee
  for $x\not\in U$, set $V_{n}(x)=0$.
 \end{enumerate}
 Thus, for $x\in U$, $V_n(x)$ is the price of the European option of maturity $\barDe_n$, with the payoff $V_{n+1}(X_{\barDe_n})$. A universal procedure suggested in \cite{QUAD03}  based on approximations of the option values and pricing kernel (called QUAD),
 is very inaccurate if the monitoring interval is small and/or the process is a VG or close to VG, e.g., KoBoL or NTS L\'evy model of order $\nu$ close to 0. The reason is evident: the pricing kernel has very large derivatives near the peak 
 (in the case of VG model, and small $\barDe$,
 non-smooth or even discontinuous), and the derivatives of $V_n, n<N,$ are very large near the barrier(s).
 See the theoretical analysis of the kernel in L\'evy models in \cite{NG-MBS}, and of the prices of the barrier options
 in \cite{DeLe14}. 
 
 Calculations in the dual space involve functions that are more regular, hence,  calculations in the dual space can be made more accurate.
 For an appropriately chosen line of integration $\{\Im\xi=\om\}$, we have
  \bbe\label{back_ind_FT}
 V_n(x)=(2\pi)^{-1}\int_{\Im\xi=\om}e^{ix\xi-\barDe_n(r+\psi(\xi))} \hV_{n+1}(\xi)d\xi.
 \ee
 Typically, $\hV_{N-1}(x)$ can be calculated in the analytical form but for $n<N-1$, the exact calculations become impossible. Naturally, numerical calculations are impossible for all points $x\in U$. One fixes a grid
 $\vec x=(x_j)_{j=1}^M$ of points on $U$, defines $V_{N}(x_j)=(\bfo_U G)(x_j)$, and calculates 
 $\cV_n$,  the array of the approximations to $V_n(x_j), j=0,1,\ldots, M,$ given $\cV_{n+1}$.
  \begin{rem}\label{choice_grid}{\rm 
 \begin{enumerate}[(a)]
 \item
 Typically, one chooses a uniformly spaced grid $x_j=x_0+j\De$; in the case of the double barrier options,
 the optimal choice is $x_0=h_-, x_M=h_+$, $\De=(h_+-h_-)/M$; in the case of  single barrier options, either $h_-$ or $h_+$ must be at the barrier, and, in the case of puts and calls, the log-strike $\ln K$ must be one of points of the grid.
 \item
 If there are two barriers, then, typically, it is impossible to construct a uniformly spaced grid which contains
 $h_-, h_+$ and $\ln K$. Hence, an approximation of $V_n$ given values of $V_n$ at the points of the grid
 is rather inaccurate: of the order of $\De$, and accurate calculations become essentially impossible.
 If $\hV_{N-1}$ is calculated analytically, then we can regard $N-1$ as the first step of the backward induction.
 Since $V_{N-1}$ is smooth at $\ln K$, (piece-wise) polynomial interpolation of arbitrary order is possible but since the interpolation errors depend on the derivatives, which can be very large, high order interpolation does not work. For examples in 
 \cite{DeLe14}, typically, the quadratic interpolation is the best choice; for applications to Asian options, whose prices are 
 more regular, the cubic interpolation is preferred \cite{LeXi12}.
 \item
CONV method \cite{CONV,CONV2} is simple: choose a positive integer $m$, define
$M=2^m-1$, $\ze=2\pi/(\ze(M-1))$ so that the Nyquist relation $\De\ze=2\pi/(M-1)$ holds, choose an appropriate
grid $\xi_j=\ze_0+(j-1)\ze, j=0,1,2,\ldots,M,$ in the dual space, and calculates
the array $\cV_n$ as the composition of FFT applied to $\cV_{n+1}$, point-wise multiplication-by-
the array $e^{-\barDe_n(r+\psi(\xi_j))}, j=0,1,2,\ldots, M,$ and application of the inverse FFT (iFFT).
However, this implies an extremely inefficient interpolation procedure for the approximation of
$V_{n+1}$, hence, large errors. See \cite{single} for the detailed analysis. 
 \end{enumerate}
 }
 \end{rem}
 A much more accurate approach introduced  to finance in  \cite{eydeland,eydeland-mahoney} is based on approximation of 
 $V_{n+1}$ by a function whose Fourier transform can be calculated explicitly; the approximation error
 can be controlled. After the approximation,
 the integrand in \eq{back_ind_FT} is given by an explicit analytical expression, and, therefore, any numerical procedure
 can be used to evaluate the integral. The scheme was studied in detail  in \cite{single,BLdouble} in the general setting - as applications, barrier options with discrete monitoring were considered and Carr's randomization was used.
 Since the latter can be interpreted as a backward procedure (with the transition operator different from
 the one in the case of barrier options with discrete monitoring), the general study in
 \cite{single,BLdouble} is applicable for options with discrete monitoring/sampling. The error control
 is different (see \cite{LeXi12,DeLe14}) but the general analysis of difficulties for application of FFT and suggested remedies (refined iFFT) is valid for options with discrete monitoring/sampling. 
 
 In the case of the piece-wise linear interpolation,  we obtain
 \bbe\label{linearint}
 \cV_{n;j}=\sum_{0\le k\le M} c^{\barDe}_{j-k}\cV_{n+1;k}+\ldots,
 \ee
 where  \bbe\label{c_ell_bDe}
 c^{\barDe}_\ell=\frac{\De}{2\pi}\int_{\bR}e^{i\De\ell}e^{-\barDe(r+\psi(\xi))}A^{\De}(\xi)d\xi
 \ee
(see \cite{single}) and
 \bbe\label{AbarDe}
 A^{\De}(\xi)=\frac{e^{i\De\xi}+e^{-i\De\xi}-2}{(i\De\xi)^2}.
 \ee
 We omit the boundary terms which are linear functions of $\cV_{n+1;k}, k=0,1,$ and $k=M-1,M$;
 the CPU cost of calculation of the omitted terms, for all $j$, is $O(M)$.  
 \begin{rem}\label{disc_barr_rem}{\rm 
 \begin{enumerate}[(1)]
 \item
 The
 apparent singularity under the integral sign is removable; formulas for the omitted terms can be found
 in \cite{single}.
 \item
 In the case of piece-wise interpolation procedures of higher order, the omitted boundary terms involve
 the function values at more than two points of the grid at and close to the boundary point, and an additional factor
$A^{\barDe}(\xi)$ in the formula
 for the coefficients $c^{\barDe}_\ell$ is more involved. See \cite{LeXi12,DeLe14}.
 \item
 The price of the barrier option with discrete monitoring is discontinuous at the barrier\footnote{
 The same is true at one of the boundaries for barrier options with continuous monitoring
 if the underlying process is of finite variation with non-trivial drift - see \cite{early-exercise,BIL}},
 therefore, in order that the interpolation error not be large, in the omitted terms, $\cV_{n+1,0}$ and $\cV_{n+1,M}$ must be replaced with the limits $V_{n+1}(x_0+0), V_{n+1}(x_M-0)$, and the formula \eq{linearint}
 applied for all $j=0,1,\ldots, M$.
 \item
 Extending $\cV_n$ by 0 to a function on $\bZ$, we can write the main block in \eq{linearint}
 as the discrete convolution operator; after an appropriate truncation, we have a finite sum.
  Identifying the grid with the group of roots of unity of order $M$,
 and using the spectral analysis in spaces of functions on this group, one obtains a useful representation of
 the discrete convolution operator in terms of the discrete Fourier transform and its inverse\footnote{In the better known
 case of functions on the real line, the convolution operator with the kernel $k$ can be represented as the composition of the Fourier transform, multiplication-by-$\hat k$-opersator, and the inverse Fourier transform.}. See, e.g.,
\cite{BriggsHenson}.   This algorithm is expressible via
 iFFT and FFT; as it was remarked in \cite{single}, while doing calculations in MATLAB,
 it may be better to use built-in procedures for iFFT and FFT rather than the one for the fast convolution.
 
 \item
 In the majority of the literature when the Fourier transform or Hilbert transform are used,
 the main block is a discrete convolution algorithm as well. 
 The entries of the discrete convolution kernel are expressed using iFT - the expressions are similar to
 \eq{c_ell_bDe}. The standard approach is to use the same pair of uniform grids in the state space and dual space
 for all purposes, with the same number of points and the steps related via the Nyquist relation.
 However, as it was explained in \cite{single}, this approach may require unnecessary large arrays even
 if the fractional FFT is used.  Accurate calculations of the integral on the RHS of \eq{linearint}
 require a very large long grid, especially in view of the fact that the discretized kernel must be calculated with the precision higher than the required precision for calculation of option values. 
 As the standard numerical analysis wisdom suggest, one has to calculate such a kernel with an accuracy much
 higher than the option values at each step. As the remedy, in \cite{single}, the refined iFFT-FFT procedure is suggested.
 The procedure allows one to (almost) independently choose the grid in the state space needed to control  approximations of value functions at each time step (this grid is, typically, short and contains a fairly small number of points, especially in
 the case of the double-barrier options), the grid in the state space needed to accurately approximate the action of
 the convolution operator in the state space, and the grid in the dual space needed for accurate evaluation of the RHS
 of \eq{linearint} for all $k$ used in the discrete convolution algorithm. 
 \item
 In the paper, as in \cite{LeXi12,DeLe14}, we will calculate the entries $\be_{a;j}$ of the discrete convolution 
 kernel using the conformal acceleration method. This allows us to independently control different sources of errors, and easily achieve
 the high accuracy of calculation of the entries.

\end{enumerate}
}\end{rem}

 \subsection{B-spline Basis Density projection method}\label{ss: Bsplines}
  In a sequence of works, the B-spline density projection method was shown to be an effective method of option pricing for vanilla and exotic options \cite{Ki14,KiDe14, Ki14B, Ki16A,CuiKirkNguyenDVSwap17}. This method is based on the theory of frames and Riesz bases, see \cite{OC03, CH11,YO80}, and in particular the use of orthogonal projection. There are at least a couple of distinct ways in which projection may be applied in a numerical pricing context. 
  In \cite{KiDe14}, the theory of Riesz bases is used to approximate the \emph{payoff} (value) function based on its orthogonal projection onto a B-spline basis. In \cite{Ki14}, the \emph{transition density} of a random variable is projected onto the basis, and used to price financial derivatives. This latter approach, called the PROJ method, is the focus of the present work, where we provide a highly accurate  method for the density projection, using the machinery of SINH-acceleration. 
  
  There are several advantages to using the B-spline basis to approximate the transition density:
 \begin{enumerate}
 \item Accuracy: the B-spline density approximations converge at a high polynomial order. While the COS method can achieve exponential convergence, for problems commonly encountered in practice (namely highly peaked transition densities) the convergence of COS is typically slower than for B-splines, and COS requires more computational effort to achieve a similar accuracy.
 \item Robustness: unlike global bases (such as COS), the local nature of B-splines makes them more robust to features of the density that can cause Gibbs oscillations. For example, the Variance Gamma process is notoriously difficult, as it exhibits in a cusp-like (non-smooth) transition density for small time horizons. Special techniques such as spectral filtering are required for COS to succeed for such densities, see \cite{ruijter2015application}.  By contrast, the B-spline coefficient formula already involves a natural spectral filter (the dual scaling function, discussed below), which increases  the decay of the characteristic function at infinity, and counteracts the occurrence of Gibbs oscillations to a large degree. 
 \item Tractability: the B-spline bases are mathematically very tractable, as they are formed from simple local polynomials with compact support. This makes it easy to extend the PROJ method to new problems in finance. Given the large range of applicability of the PROJ method, a significant improvement in its computation is beneficial to many applications.
 \end{enumerate}
 
 Let timestep $\barDe >0$ be such that the transition density  $p_{\barDe}\in L^2(\mathbb R)$. Since $\norm{p_{\barDe}}^2_2\leq \norm{p_{\barDe}}_\infty$, all bounded densities are in $L^2(\mathbb R)$.\footnote{Note however, that in the VG model, for sufficiently small $\barDe$, $\norm{p_{\barDe}}_\infty=\infty$.} 
The idea in \cite{Ki14} is to project  $p_{\barDe}$ onto a non-orthogonal basis, generated by a \emph{scaling function} $\varphi$, which is compactly supported and symmetric. Let $\De>0$, and $x_k=x_0+k\Delta$, $k\in \mathbb Z$, where $x_0\in \mathbb R$ is a shift parameter. By shifting and scaling $\varphi$ at a \emph{resolution} $a:=1/\De$, we form the  basis $\{\varphi_{a,k}(x)\}_{k\in \mathbb Z}:=\{a^{1/2}\varphi(a(x-x_k)) \}_{k\in \mathbb Z}$.
As long as
\begin{equation}\label{eq: frameb}
A\norm{f}_2^2\leq\sum_{k\in \mathbb Z} |\langle f,\varphi_{a,k}\rangle|^2\leq B \norm{f}_2^2, \quad \forall f\in L^2(\mathbb R),
\end{equation}
for some $0<A\leq B$ called the \emph{frame bounds}, $\varphi$ is said to generate a \emph{Riesz basis} for its closed span $\mathcal M_a:=\cspan{\varphi_{a,k}}_{k\in \mathbb Z}$. A Riesz basis has the property that every $f\in \mathcal M_a$ is uniquely representable, see Theorem 3.6.3 of \cite{OC03}. Let $P_{\mathcal M_a}$ be the orthogonal projection operator on $\mathcal M_a$. The coefficients in the unique 
expansion 
\begin{equation}\label{eq:ProjGen}
P_{\mathcal M_a}p_{\barDe}(x) = \sum_{k\in \mathbb Z} \beta_{a,k}\cdot\varphi_{a,k}(x),
\end{equation}
 are called projection coefficients. From Theorem 3.6.3 and Lemma 7.3.7 of \cite{OC03}, there exists a \emph{dual generator} $\widetilde \varphi$ such that
\[
\beta_{a,k}=\langle p_{\barDe},\widetilde{ \varphi}_{a,k}\rangle= \int_{-\infty}^\infty p_{\barDe}(x) \overline{\widetilde{ \varphi}_{a,k}(x)}dx.
\]
By Proposition 7.3.8 of \cite{OC03},
\begin{equation}\label{eq: dualft}
\widehat{\widetilde \varphi}(\xi)=\frac{\widehat \varphi(\xi)}{ \mathbf \Phi(\xi)},   \quad \xi\in \mathbb R,
\end{equation}
where $\mathbf \Phi(\xi):=\sum_{k\in \mathbb Z}\left|\widehat\varphi(\xi +2\pi k) \right|^2$. Moreover, by Theorem 7.2.3 of \cite{OC03}, $0 < A \leq \mathbf \Phi(\xi) \leq B$, for $\xi\in \mathbb R$. Using these facts, 
\cite{Ki14} derives 
\bbe\label{beakgen}
\beta_{a,k}=\frac{a^{-1/2}}{2\pi}\int_\bR e^{-ix_k\xi-\barDe\psi(\xi)}\widehat{\widetilde \varphi}\Big(\frac{\xi}{a}\Big)d\xi,
\ee
as well as
 explicit formulas for B-spline bases of various orders. 
 For the linear B-splines,
 \bbe\label{htphilin}
\htphi(\xi)=\frac{12\sin^2(\xi/2)}{\xi^2(2+\cos\xi)}=\frac{6(1-\cos\xi)}{\xi^2(2+\cos\xi)}.
\ee
 The
 formulas for the splines of higher order are in Sect. \ref{coeffquadratic}, \ref{coeffcubic}, \ref{coeffhigher}. Note that $\widehat{\widetilde \varphi}\Big(\frac{\xi}{a}\Big)$ in \eq{beakgen} provides a natural \emph{spectral filter}, which increases the rate of the decay of the integrand at infinity.
 
 The rates of convergence of the  orthogonal projection to the true density, in several norms, are given in the following propositions. 
 \begin{prop}[\cite{Un97}, Theorem 4.3]\label{props2} If for $m=0,...,L-1$, $\widehat \varphi^{(m)}(2\pi k) = 0$, $k\in \mathbb Z$, $k\neq 0$, then there exists 
a positive constant $K_{\varphi,L}$ s.t. $\forall f\in W_2^L$
\begin{equation}
\norm{f - P_{ \mathcal M_a} f}_2 \leq K_{\varphi,L}\cdot \Delta^L \cdot \norm{f^{(L)}}_2,
\end{equation}
where $\mathcal M_a = \mathcal M_a(\varphi):=\cspan{\varphi_{a,k}}_{k\in \mathbb Z}$.  For functions which satisfy the stricter smoothness requirement $f\in W_2^{2L}$, we have
\begin{equation}
\norm{f - P_{ \mathcal M_a} f}_2 \leq K_{\varphi,2L} \cdot \Delta^{2L} \cdot \norm{f^{(2L)}}_2 + K_{\varphi,2L}^{1/2} \cdot \Delta^L \cdot \norm{f^{(L)}}_2.
\end{equation}
for some constant $ K_{\varphi,2L}$.
\end{prop}
We also have the $L_\infty$ bound in the Sobolev space $W_\infty^L$ of functions whose first $L$ deriviates are defined in the $L_\infty$ sense, where $\norm{f}_\infty:=\sup_{x\in \mathbb R}|f(x)|$.
\begin{prop}[\cite{Un97}, Proposition 3.3]
If $\varphi$ satisfies the conditions of Proposition \ref{props2}, then $\exists C_{\varphi,\infty}$ such that $\forall f\in W^L_\infty$,
\begin{equation}
\norm{f - P_{ \mathcal M_a} f}_\infty \leq C_{\varphi,\infty}\cdot \Delta^L \cdot \norm{f^{(L)}}_\infty.
\end{equation}
\end{prop} 
In the context of the approximation of the probability kernel, we need the error bound in $L_1$-norm,
because this bound gives the bound for the resulting contribution to the error of the evaluation of the value
function at each time step, in $L_\infty$-norm. In the case of single barrier options with unbounded payoffs,
one can apply this bound after an appropriate change of measure, which reduces the pricing problem to the case of options with uniformly bounded payoffs.

\begin{prop}\label{BsplinesL1err}
Let $f\in L^1(\mathbb R;\bR_+)$, and let $a>0$ be sufficiently large so that $\norm{f-P_{ \mathcal M_a} f}_\infty<1$.
Then, for any $\kappa>0$,
\bbe\label{errBLinfty}
\norm{f-P_{ \mathcal M_a}}_{1}\le \int_{|x|>\kappa}f(x)dx + \kappa \norm{f-P_{ \mathcal M_a} f}_\infty.
\ee
\end{prop}
\begin{proof} Since $f\ge 0$ and $\norm{f-P_{ \mathcal M_a} f}_\infty<1$,
we have $|f(x)-P_{ \mathcal M_a} f(x)|\leq |f(x)| = f(x)$ for all $x \in\bR$, and, therefore,
\begin{align*}
\int_\mathbb R|f(x)-P_{ \mathcal M_a} f(x)|dx& \leq \int_{|x|>\kappa}|f(x)-P_{ \mathcal M_a} f(x)|dx +  \int_{|x|\leq \kappa}|f(x)-P_{ \mathcal M_a} f(x)|dx  \\
& \leq \int_{|x|>\kappa}|f(x)|dx +  \int_{|x|\leq \kappa}(f(x)-P_{ \mathcal M_a} f(x))dx.
%  \\& \leq \int_{|x|>\kappa}f(x)dx + 2\kappa \norm{f-P_{ \mathcal M_a} f}_\infty.
\end{align*}
\end{proof}
\begin{rem}\label{rem:Kernel_approx}{\rm  The bounds above can be directly applied if the payoff function is bounded.
This is the case for the put option, the up-and-out options, and options with two barriers. In the case of down-and-out call options, after the truncation to an interval $[h, b]$, one works with bounded functions, which may have 
a very large $L_\infty$-norm, denote it $A$. If the right tail of the L\'evy density  decays slowly, as $e^{\lm x}$, where $\lm<-1$ is close to 1, then $b$ must be rather large and then $A$ is of the order of $e^b$, which is very large indeed. In this case, the error $\eps$ in the $L_1$-norm in the approximation
of the kernel introduces the error of the order of $\eps A$, at each time step. Hence, the number of terms in
an accurate 
B-spline approximation must be much larger than Proposition \ref{BsplinesL1err} implies.
Hence, if the truncation parameter $b$ is very large, one should make an appropriate change of measure
to reduce to the case of an option with the payoff function of the class $L_\infty$. The change of measure will change
the pricing kernel and constants in the error bounds for the B-spline approximation but the errors will be easier to control.
}
\end{rem}

 \subsection{Approximation errors and spectral filter errors}\label{appr_spectral_filter}
 One can interpret an additional factor $A^{\De}(\xi)$ in \eq{c_ell_bDe} as a spectral filter
 which improves the convergence of the iFT representation of the pricing kernel; equivalently,
 the initial pricing kernel is replaced with a smoother one. The additional factor  in \eq{beakgen} below admits a similar interpretation. Spectral filtering is a popular tool in engineering - see, \cite{GasWit98}. Notice, however, that in both cases
 \eq{c_ell_bDe} and \eq{beakgen},  the function interpreted as a spectral filter arises as the result of
 a certain approximation, and the error of approximation is controlled. By contrast, \cite{ruijter2015application}
 uses ad-hoc spectral filters to increase the convergence of the integrals:
``When Fourier techniques are employed to specific option pricing cases from computational
finance with non-smooth functions, the so-called Gibbs phenomenon may become apparent.
This seriously impacts the efficiency and accuracy of the pricing. For example, the Variance
Gamma asset price process gives rise to algebraically decaying Fourier coefficients, resulting
in a slowly converging Fourier series. We apply spectral filters to achieve faster convergence.
Filtering is carried out in Fourier space; the series coefficients are pre-multiplied by a decreasing
filter, which does not add significant computational cost. Tests with different filters show how
the algebraic index of convergence is improved."

Although the quoted statement is correct but the implications for applied finance are ignored. Indeed,
spectral filters are designed to regularize the results. In applications to derivative pricing, the regularization
results in serious errors in regions of the paramount importance for risk management: near barrier and strike,
close to maturity and for long dated options. The errors of CM and COS methods in calibration procedures documented in
the extensive numerical study in \cite{HestonCalibMarcoMe,HestonCalibMarcoMeRisk}, and for pricing barrier options
in \cite{DeLe14} (the errors of COS may blow up starting with maturities 0.5Y) are artifacts of such filtering.
The  spectral filters implicit in \eq{c_ell_bDe} and \eq{beakgen} introduce errors as well but this error is the error of
the approximation of the value functions and pricing kernel, respectively, and can be controlled efficiently. One of
the purposes of the paper is an experimental study of the relative accuracy and efficiency of the spectral filters 
implicit in \eq{c_ell_bDe} and \eq{beakgen}. Note that PROJ method involves two approximations, hence, in effect, two spectral filters, and it is interesting to study whether and when and where additional filtering improves the performance of the numerical scheme. For example, \cite{CuiKirkNguyenCliquet17} find that the application of an additional (exponential) filter helps to smooth the convergence for cliquet-style contracts in the presence of capped/floored payoffs.

 A general characterization of the relative advantages/disadvantages of the method in \cite{DeLe14} vs.~the B-spline projection method is as follows. 
 
 \begin{enumerate}[(1)]
 \item
 {\sc Similarities.} Both methods utilize the truncation in the state space and approximation of
 the value functions by piece-wise polynomials. The parametrizations of the piece-wise polynomials is different
 and the choice is determined by the convenience considerations. On the fundamental level, there is no difference.
 \item
 {\sc Slight differences.} In \cite{DeLe14}, more accurate error bounds are derived and recommendation for the truncation and discretization parameters are given. In many cases, these recommendations lead which to an overkill and larger CPU time than necessary to satisfy the given error tolerance;
 the error bounds in the papers where the B-spline projection method is used are less explicit, and recommendations are ad-hoc.
 \item
 {\sc Fundamental differences.} The B-spline approximation of the pricing kernel introduces an additional error
 which the method in \cite{DeLe14} does not have. The error can be controlled if  the probability density
 is sufficiently regular; if the probability density has a very high peak or very heavy tails, the error control becomes inefficient, and if
 the probability density is non-smooth or discontinuous, the B-spline method becomes an informal spectral filter which 
 regularizes the irregular correct result. The advantage of the B-spline approximation is the increase of speed
 of calculations because the approximate pricing kernel decays faster at infinity, hence,
 calculations become faster.
 \end{enumerate}
 % \subsection{Drawbacks of FFT method and remedies}\label{Pitfalls_remedies}

% \subsection{SINH-acceleration}\label{ss:SINHacceleration}

 \section{Calculation of elements $\be_{a;k}$ using the SINH-acceleration technique}\label{s:SINH-PROJ}
 \subsection{The scheme} For simplicity, we assume that $\psi$ satisfies conditions \eq{psi0}, \eq{asymppsi} and \eq{asymppsipos} with $\cC'=\cC_{-\ga_\nu,\ga_\nu}$. As we mentioned in Sect. \ref{ss:conf}, these conditions are satisfied for essentially all models popular in finance bar
 the VG model. For the latter, the factor $|\xi|^\nu$ on the RHS of  \eq{asymppsi} should be replaced with $\ln|\xi|$.
 The calculations below remain essentially the same, as well as the approximate bound for the discretization
 error and recommendation for the choice of the step $\ze$. The choice of the truncation parameters, which depend on the rate of decay of $\psi$ at infinity, should be modified in the same vein as in \cite{iFT0,iFT,BoyarLevenSinh19}.
 The conditions \eq{asymppsi} and \eq{asymppsipos} can be relaxed, and any sinh-regular process (see \cite{BoyarLevenSinh19}) used. 
 
 First, we study in detail the case of linear  splines, hence, for $\htphi(\xi)$ given by 
 \eq{htphilin}; modifications for splines of higher order are in Sect. \ref{coeffquadratic}-\ref{coeffhigher}. Changing the variable $\xi\mapsto a\xi$, we obtain
 \bbe\label{be_via_I}
 \beta_{a,k}=a^{1/2}I(x_k; \De; \barDe; \psi),
 \ee
 where 
 \bbe\label{Iaxk1}
I(x_k; \De; \barDe; \psi)=\frac{1}{2\pi }\int_\bR e^{-ix_k a\xi-\barDe\psi(a\xi)}\htphi(\xi)d\xi.
\ee
For $a$ fixed, the characteristic exponent $\xi\to \psi(a\xi)$ satisfies  \eq{psi0}, \eq{asymppsi} and \eq{asymppsipos}
with the same $\nu$ and  coni as $\psi$ does but the strip of analyticity is $S_{(\mum/a,\mup/a)}$ instead of $S_{(\mum\mup)}$,
and $c_{\infty;a}=a^\nu c_\infty$ instead of  $c_\infty$. 

The integral on the RHS of \eq{Iaxk1}
can be regarded as the price of the European option of maturity $\barDe$, with the payoff whose Fourier transform
is $\htphi(\xi)$, in exponential L\'evy models. A 
seemingly important difference with the case of puts and calls is that the Fourier transforms of the
payoffs $(e^x-K)_+$ and $(K-e^x)_+$ admit meromorphic continuation to the complex plane with simple poles
at $0, -i$, whereas $\htphi(\xi)$ has infinite number of poles. First, we will calculate the poles, and then
show that the integrals $I(x_k; \De; \barDe; \psi)$ can be calculated almost as easily as the European puts and calls priced.

\subsection{Calculation of the poles of the integrand}\label{calc_of_poles}
\begin{lem}
\begin{enumerate}[a)]
\item
$\htphi(\xi)$ is meromorphic in the complex plane. All poles are simple, and they are of the form  
\bbe\label{xiellpm}
\xi^\pm_\ell=(2\ell+1)\pi\pm i \ln(2+\sqrt{3}), \ \ell\in\bZ.
\ee
\item
For any $\ga\in (0,\pi/2)$, $\htphi(\xi)\sim -6\xi^{-2}$ as $\xi\to \infty$ in 
$\{\xi\ |\ \mathrm{arg}\,\xi\in (\pi/2-\ga, \pi/2+\ga)\}\cup \{\xi\ |\ \mathrm{arg}\,\xi\in (-\ga-\pi/2, -\pi/2+\ga)\}$.
\end{enumerate}
\end{lem}
\begin{proof} a) The apparent singularity at 0 is, in fact, removable, hence, any pole is a solution of the equation 
$4+w(\xi)+1/w(\xi)=0$, where $w(\xi)=e^{i\xi}$. We find $w=-(2+\sqrt{3})^{\pm 1}$, and conclude that all poles are simple
and of the form \eq{xiellpm}.

b) If $\xi\to \infty$ in  the cone $\{\xi\ |\ \mathrm{arg}\,\xi\in (\pi/2-\ga, \pi/2+\ga)\}$, then $w(\xi)\to 0$ and
\[
\htphi(\xi)=\frac{6(2-w(\xi)-1/w(\xi))}{\xi^2(4+w(\xi)+1/w(\xi))}=-\frac{6(1-2w(\xi)-w(\xi)^2)}{\xi^2(1+4w(\xi)+w(\xi)^2)}\sim -6\xi^{-2}.
\]
The case $\xi\to \infty$ in  the cone $\{\xi \ |\ \mathrm{arg}\,\xi\in (-\ga-\pi/2, -\pi/2+\ga)\}$ is similar. 
\end{proof}
As in the case of  puts and calls, we separate the oscillating factor in the integrand in \eq{Iaxk1}:
\bbe\label{Iaxk2}
I(x_k; \De; \barDe; \psi)=\frac{1}{2\pi }\int_\bR e^{ix'_k a\xi-\barDe\psi^0(a\xi)}\htphi(\xi)d\xi,
\ee
where $x'_k=-x_k+\mu\barDe$, and, depending on the sign of $x'_k$, deform the line of integration either upward (if $x'_k\ge 0$) or downward
(if $x'_k\le 0$);  if $x'_k=0$, we can deform the line of integration in either direction or not at all. 
We will evaluate $I(x_k; \De; \barDe; \psi)$ efficiently using appropriate contour deformations. 
Efficient procedures require crossing the poles, hence, we will have to use the residue theorem and
calculate the residues $Res(x_k; \De; \barDe; \psi; 
\xi^\pm_\ell)$ of the integrand at $\xi^\pm_\ell$, $\ell\in \bZ$. 
 \begin{lem}  For $\ell\in \bZ$,
 \bbe\label{Resellp}
 Res(x_k; \De; \barDe; \psi; 
\xi^\pm_\ell)=\pm A^\mp(x'_0,a)(-1)^{k\, \mathrm{mod}\, 2}(2+\sqrt{3})^{\mp k}
\frac{e^{ix'_0a2\ell\pi-\barDe\psi^0(a\xi^\pm_\ell)}}{(\xi^\pm_\ell)^2},
\ee
where $A^\mp(x'_0,a)= \frac{36 e^{ x'_0a(\mp\ln(2+\sqrt{3})+i\pi)}}{(2+\sqrt{3})-(2+\sqrt{3})^{-1}}$, and
\beqa\nonumber
SR^\pm(x'_0, a, k)&:=&\sum_{\ell\in \bZ} Res(x_k; \De; \barDe; \psi; 
\xi^\pm_\ell)\\\label{sumRes}
&=&\pm A^\mp(x'_0,a)(-1)^{k\, \mathrm{mod}\, 2}(2+\sqrt{3})^{\mp k}J^\pm(x'_0,\barDe,a),
\eqa
where
\bbe\label{JbarDea}
J^\pm(x'_0,\barDe,a)=\sum_{\ell\in\bZ} \frac{e^{ix'_0a 2\ell \pi-\barDe\psi^0(a\xi^\pm_\ell)}}{(\xi^\pm_\ell)^2}
\ee
is independent of $k$.

 \end{lem}
 \begin{proof} %Consider the residue at $\xi^+_\ell$. 
 Since $w(\xi^\pm_\ell)+1/w(\xi^\pm_\ell)=-4$, we have
 \[
 2-w(\xi^\pm_\ell)-1/w(\xi^\pm_\ell)=6. 
 \]
 Then, 
 \beqast
 \frac{d}{d\xi}(4+w(\xi)+1/w(\xi))\vert_{\xi=\xi^\pm_\ell}
 &=&\left((1-w(\xi)^{-2}) \frac{dw(\xi)}{d\xi}\right)\vert_{\xi=\xi^\pm_\ell}\\
&=&i \left((1-w(\xi)^{-2})w(\xi)\right)\vert_{\xi=\xi^\pm_\ell}\\
 &=&i(-(2+\sqrt{3})^{\mp 1}+(2+\sqrt{3})^{\pm 1})\\
 &=&\pm i(2+\sqrt{3}-1/(2+\sqrt{3})).
 \eqast
 Next, since $a\De=1$,
 \beqast
 e^{ix'_ka\xi^\pm_\ell}&=&e^{i(x'_0+k\De)a((2\ell+1)\pi\pm i\ln(2+\sqrt{3}))}\\
 &=&e^{\mp x'_0a\ln(2+\sqrt{3})}(2+\sqrt{3})^{\mp k}(-1)^{k(2\ell+1)}e^{ix'_0a(2\ell+1)\pi},
 \eqast
 and \eq{Resellp}- \eq{sumRes} follow.

 \end{proof}
 
\subsection{The contour deformations in the atypical case}
 Typically, $\mup>0, \mum<-1$ are not very large in
absolute value, and $\De<1$ (even $<<1$), hence, $a>1$ (even $>>1$). Then the
strip of analyticity of the function $\xi\mapsto\psi^0(a\xi)$ contains no poles of $\htphi$. However,
for the simplicity of exposition,  we consider first the case when all the poles $\xi^\pm_\ell$ belong to the strip,
and indicate the changes needed in the other cases later.

%%%%%%%%%%%%%%%%%%%%%%%%%%%%%%%%%
\subsubsection{Crossing the poles}
%%%%%%%%%%%%%%%%%%%%%%%%%%%%%%%%%
\begin{lem}\label{BSINHcrossp}
a) Let $x'_k\ge 0$, $\ln(2+\sqrt{3})< \mup/a$, and $\om_0\in (\ln(2+\sqrt{3}), \mup/a)$. Then 
\bbe\label{Intcrossp}
I(x_k; \De; \barDe; \psi)=SR^+(x'_0, a, k)+\frac{1}{2\pi }\int_{\Im\xi=\om_0} e^{ix'_k a\xi-\barDe\psi^0(a\xi)}\htphi(\xi)d\xi.
\ee
b) Let $x'_k\le 0$, $-\ln(2+\sqrt{3})> \mum/a$, and $\om_0\in (\mum/a,- \ln(2+\sqrt{3}))$. Then 
\bbe\label{Intcrossm}
I(x_k; \De; \barDe; \psi)=-SR^-(x'_0, a, k)+\frac{1}{2\pi }\int_{\Im\xi=\om_0} e^{ix'_k a\xi-\barDe\psi^0(a\xi)}\htphi(\xi)d\xi.
\ee
\end{lem}

\begin{proof} a) For $N\in \bN$, introduce $U_N:=\{z\in \bC\ | \ 0\le \Im z\le \om_0, |\Re z|\le 2N\pi\}$ and the integral
\[
I(x_k; \De; \barDe; \psi; U_N)=\frac{1}{2\pi}\int_{\dd U_N} e^{ix'_k a\xi-\barDe\psi^0(a\xi)}\htphi(\xi)d\xi.
\]
The integrand is meromorphic in $U_N$, with the simple poles $\xi^+_\ell, |\ell|\le N-1$. By the residue theorem, 
\[
I(x_k; \De; \barDe; \psi; U_N)=\sum_{|\ell|\le N-1}Res(x_k; \De; \barDe; \psi; 
\xi^+_\ell).
\]
Since the integrals over the vertical sides of the rectangle $U_N$ decay as $N^{-2}$ as $N\to+\infty$, and 
$Res(x_k; \De; \barDe; \psi; 
\xi^+_\ell)=O(\ell^{-2})$, as $\ell\to\pm\infty$, we can pass to the limit and obtain
\[
I(x_k; \De; \barDe; \psi)-\frac{1}{2\pi }\int_{\Im\xi=\om_0} e^{ix'_k a\xi-\barDe\psi^0(a\xi)}\htphi(\xi)d\xi=SR^+(x'_0, a, k).
\]
b) This time, $U_N:=\{z\in \bC\ | \ \om_0)\le \Im z\le 0, |\Re z|\le 2N\pi\}$, and 
\[
\frac{1}{2\pi }\int_{\Im\xi=\om_0} e^{ix'_k a\xi-\barDe\psi^0(a\xi)}\htphi(\xi)d\xi-I(x_k; \De; \barDe; \psi)=SR^-(x'_0, a, k).
\]
\end{proof}
Thus, to calculate the sums of the series of residues at $\xi^\pm_\ell$, it suffices to calculate appropriate partial 
 sums of the series $J^\pm(x'_0,\barDe,a)$ (see Section \ref{choiceN}) which are independent of $k$, and then multiply the result by a simple expression depending
 on $k$. If $\barDe$ is not very small, the  sums are easy to calculate with high precision if $\nu\ge 1$ because
 then the terms decay exponentially or super-exponentially. If $\nu\in (0,1)$ and $e^{i2x_0a\pi}$ is not very close to 1, the rate of convergence
 of the series can be improved using the summation by parts as explained in \cite{Contrarian}.
 In any case, the series $J^\pm(x'_0,\barDe,a)$ need to be calculated only once, hence, the CPU time cost is small.
 
 It remains to calculate the integrals on the RHS' of \eq{Intcrossp} and \eq{Intcrossm}.
 We assume that $\psi^0$ admits analytic continuation to $i(\mum,\mup)+\cC_{-\ga_\nu,\ga_\nu}$, where $\mum<0<\mup$, 
 $\ga_\nu=\min\{\pi/2, \pi/(2\nu)\}$, $\nu\in (0,2]$, and satisfies conditions\eq{psi0}, \eq{asymppsi} and \eq{asymppsipos} with $\cC'=\cC_{-\ga_\nu,\ga_\nu}$.
 \subsubsection{Application of SINH-acceleration in the case $x'_k\ge 0$}
%%%%%%%%%%%%%%%%%%%%%%%%%%%%%%%%%
 The integrand on the RHS of \eq{Intcrossp} is analytic in $i(\ln(2+\sqrt{3}), \mup/a)+\cC_{0,\gap}$, and decays at infinity as
 the integrand in the pricing formula for the OTM European put \cite{BoyarLevenSinh19}
 in KoBoL model,
 with $(\ln(2+\sqrt{3}), \mup/a)$ in place of $(0,\lp)$. Hence, we can use  the recommendation in  \cite{BoyarLevenSinh19}
 to choose the parameters of the sinh-change of variables
 \begin{equation}\label{sinhbasic}
\xi=\chi_{\om_1, \om; b}(y)=i\om_1+b\sinh (i\om+y),
\end{equation}
where $\om_1\in \bR, \om\in (0,\pi/2)$ and $b>0$ are related as follows: $\om_0=\om_1+b\sin (\om)$.
The Cauchy integral theorem allows us to deform the line of integration 
$\{\Im\xi=\om_0\}$ into the contour $\cL_{\om_1,\om; b}:=\chi_{\om_1, \om; b}(\bR)$. In the integral over
$\cL_{\om_1, \om; b}$, we make the change of variables \eqref{sinhbasic}:
\beqast
I^+(x_k; \De; \barDe; \psi)&=& \frac{1}{2\pi }\int_{\Im\xi=\om_0} e^{ix'_k a\xi-\barDe\psi^0(a\xi)}\htphi(\xi)d\xi\\
&=&\frac{1}{2\pi }\int_{\cL_{\om_1, \om; b}} e^{ix'_k a\xi-\barDe\psi^0(a\xi)}\htphi(\xi)d\xi\\
&=&\int_\bR f(y)dy,
\eqast
where $f(y)=f(x_k; \De; \barDe; \psi;y)$ is given by
\[
f(y)=\frac{b}{2\pi}e^{ix'_k a\xi(y)-\barDe\psi^0(a\xi(y))}\htphi(\xi(y))\cosh(i\om+y),\ \xi(y)=\chi_{\om_1, \om; b}(y).
\]
With a correct choice of the parameters, the integrand is analytic in a strip $S_{(-d,d)}=\{y\ |\ \Im y\in (-d,d)\}$, $d>0$,
around the line of integration, and decays fast at infinity. Therefore, we can apply the infinite trapezoid rule
and choose a spacing of the uniform grid sufficient to satisfy the desired error tolerance using the error bound for the infinite trapezoid rule.

%%%%%%%%%%%%%%%%%%%%%%%%%%%%%%%%%
\subsubsection{Choice of the
 parameters of the sinh-acceleration and  simplified trapezoid rule}
%%%%%%%%%%%%%%%%%%%%%%%%%%%%%%%%%
We set $\om=\gap/2$, $d_0=\om$, reassign $\mum:=0, \mup:=\mup/a$, and define 
%$\om_1$ and $b$ as in \cite{BoyarLevenSinh19}: 
 %$\om_1=(\mup+\mum)/2, 
 $ b_0=(\mup-\mum)/(\sin(\gap)-\sin(\gam))$, $\om_1=\mup-b_0\sin\gap$, choose $k_d, k_b\in (0.8, 0.95)$ and set $b=k_bb_0$, $d=k_d d_0$.
  The approximate bound for the Hardy norm used in \cite{BoyarLevenSinh19}.
\[
H(f;d)=10(|f(-id)|+|f(+id)|)
\]
should be modified to take into account that the contour of integration $\cL_{\om_1, \om; b}$ passes close to the poles
$\xi^+_{\pm 1}$ if $k_d<1$ is close to 1. We use an approximate bound via
\bbe\label{apprHBsinh}
H(f;d)=10(|f(-id)|+|f(+id)|+H^+_1(f,d)+H^-_1(f,d)),
\ee
where 
\[
H^\pm_1(f,d)=\frac{|Res(x_k; \De; \barDe; \psi; \xi^+_1)|}{(\gap-\gam)(1-k_d)\pi/2}.
\]
To satisfy a small error tolerance $\eps>0$,
we choose 
\[
\ze =2\pi d/(\ln(H(f,d)/\eps).
\] 

\begin{rem}{\rm If the vectorization as in MATLAB
is used, then it is advantageous to use the minimal $\ze$ among $\ze=\ze(x'_k)$ calculated above,
and choose the truncation parameter for the smallest $x'_k$. Then the same grid can be used when the bulk of calculations
is being made (recall that the integrand, hence, the terms in the simplified trapezoid rule depend on $x'_k$
via the factor $e^{ix'_ka\xi}$ only.
However, if no vectorization is used, then it is advantageous to do the calculations point by point because
for all $x'_k\ge 0$ bar the one closest to 0, the product $e^{ix_ka\xi}$ decays much faster the smallest one, and the number of terms needed to satisfy 
the chosen error tolerance is much smaller.}
\end{rem}

We  apply the infinite trapezoid rule
\[
\int_\bR f(y)dy=\ze \sum_{j\in\bZ} f(j\ze)=2\ze\Re(0.5f(0)+\sum_{j\ge 1}f(j\ze))
\]
(the last equality is valid because $\Im{f(-y)}=-\Im f(y)$, $\Re f(y)=f(y), y\in\bR$). 
We
truncate the sum using the prescription in \cite{BoyarLevenSinh19}. 
Similarly to Eq. (2.24), in \cite{BoyarLevenSinh19}, we choose $\La_1$ to satisfy
\bbe\label{truncerror3}
Err_{tr}(\La_1)\le \frac{e^{\barDe C_0}}{\pi}\int_{\La_1}^{+\infty} e^{-(x'_k\sin\om)a\rho-\barDe c_\infty(0)\cos(\om\nu)(a\rho)^\nu}\rho^{-2}d\rho,
\ee
where 
\begin{enumerate}[(i)]
\item
 in the case of KoBoL without the BM component,\\ $C_0=c\Ga(-\nu)[\lp^\nu+(-\lm)^\nu],\ c_\infty(0)=-2c\Ga(-\nu)\cos(\pi\nu/2)$;
 \item
in the case of NTS processes of order $\nu$,   and 
$C_0=\de(\al^2-\be^2)^{\nu/2}$, $c_\infty(0)=\de$;
\item
if the BM component is not very small, $\nu=2$, $C_0=0$, $c_\infty=\sg^2/2$.
\end{enumerate}
We simplify \eq{truncerror3}: given $\eps>0$, we set $\eps_1=\pi\eps e^{-\barDe C_0}/a$, and find $\La_1>0$ satisfying
 \bbe\label{truncerror4}
e^{-x'_k \sin(\om)(a\La_1)-\barDe c_\infty(0)\cos(\om\nu)(a\La_1)^\nu}(a\La_1)^{-1}\le \eps_1.
\ee
We set $a_1=x'_k\sin\om, a_2=\barDe c_\infty(0)\cos(\om\nu)$, $C_1=\ln(1/\eps_1)$ and solve
the equation
\[
g(\La_2):=a_1\La_2 +a_2\La_2^\nu+\ln\La_2-C_1=0.
\]
\begin{enumerate}[(a)]
\item
If $\nu< 1$, it is easy to check that $g'>0$ and $g^{\prime\prime}<0$ on $\bR$, $g(0^+)=-\infty, g(+\infty)=+\infty$.
Hence, $\La_2$ can be easily found using Newton's method with the initial approximation
$\La^0_2:=\min\{C_1/(3a_1), (C_1/(3a_2)^{1/\nu}, e^{C_1/3}\}$.

\item
If $\nu>1$, we make the change of variables $\La_2=\La_3^{1/\nu}$, and solve the equation
\[
g_1(\La_3):=a_1\La_3^{1/\nu} +a_2\La_3+\frac{1}{\nu}\ln\La_3-C_1=0.
\]
 using Newton's method, and set $\La_2=\La_3^{1/\nu}$. 

\item
If $\nu=1$, we solve $g(\La_2)=(a_1+a_2)\La_2+\ln\La_2-C_1=0$ using Newton's method.
\end{enumerate}
The simplified version, which can be used for any $\nu\in (0,2]$, is $\La_2=(C_1/(a_1+a_2))^{\max\{1,1/\nu\}}$.

Finally, the universal simplest version, which can be used for all $x'_k\ge 0$, is to apply the simplified
prescription for the smallest $x'_k$ but use it for all $x'_k$.

When $\La_2$ is found, we calculate  $\La=\ln(2\La_2/(ab))$. 

%%%%%%%%%%%%%%%%%%%%%%%%%%%%%%%%%
\subsubsection{Summation of the series \eq{JbarDea}}\label{choiceN}
%%%%%%%%%%%%%%%%%%%%%%%%%%%%%%%%%
 We need to calculate a partial sum
 \bbe\label{JbarDeaN}
J^\pm(x'_0,\barDe,a)\approx \sum_{|\ell|\le N} \frac{e^{ix'_0a 2\ell\pi-\barDe\psi^0(a\xi^\pm_\ell)}}{(\xi^\pm_\ell)^2},
%=2\Re\sum_{\ell\ge 0}(1-\de_{\ell 0}/2) \frac{e^{ix'_0a 2\ell\pi-\barDe\psi^0(a\xi^\pm_\ell)}}{(\xi^\pm_\ell)^2},
\ee
where %$\de_{jk}$ is Kronecker delta, 
$N$ is chosen to satisfy the desired error tolerance $\eps$.  We set $k_+=\min\{k\ |\ x'_k\ge 0\}$,
 \[
 \eps_1=\frac{\eps (2+\sqrt{3})^{k_+}}{a|A^\mp(x'_0,a)|}, 
 \]
and choose $N$ to satisfy
\bbe\label{truncSR}
\left|e^{-\barDe\psi^0(a(2N+1+i(2+\sqrt{3}))}\right| |a(2N+1+i(2+\sqrt{3}))|^{-1}\le \eps_1.
\ee
Assuming that $a$ is large, we can use a simplified prescription 
\bbe\label{truncSR2}
\left|e^{-\barDe c_\infty(0)(a(2N+1+i(2+\sqrt{3}))^\nu}\right| (a(2N+1))^{-1}\le \eps_1.
\ee
As in the case of the choice of the truncation parameter above, we find an approximate $N$ 
solving approximately the equation
$
\barDe c_\infty(0) \La^\nu +\ln\La - \ln(1/\eps_1)=0.
$
We change the variable $\La_1=\La^\nu$, and solve the equation
\[
g(\La_1)=\La_1+a_1\ln (\La_1)-C_1=0,
\]
where $a_1=1/(\nu \barDe c_\infty(0)), C_1=\ln(1/\eps_1)/(\barDe c_\infty(0))$ using Newton's method. Then we set
$\La=\La_1^{1/\nu}, N=\mathrm{ceil}\, (\La/(2a))\}$.
If $\nu$ is small, $N_1$ can be uncomfortably large. However, if $\nu\in (0,1)$ and $e^{ix'_0a2\pi}$ is not close to
1, we can use the summation by parts
Boyarchenko \& Levendorski\u{i} (2019b), and significantly decrease $N_1$ necessary to satisfy the desired error
tolerance.

%\begin{rem}{\rm This is the place where the choice of $x'_0$ might  matter but it is not because $e^{i\De a 2\pi}=1$.}
%\end{rem}
%%%%%%%%%%%%%%%%%%%%%%%%%%%%%%%%%
 \subsection{The contour deformations in the case $x'_k \ge 0$, $\mup/a\le \ln(2+\sqrt{3})$ or 
 $\mup/a$ is very close to $\ln(2+\sqrt{3})$} 
%%%%%%%%%%%%%%%%%%%%%%%%%%%%%%%%%
We take $\om_0>\ln(2+\sqrt{3})$ and  deform the line of integration $\bR$ into
 the line $\{\Im\xi=\om_0\}$, whose part above $[-\pi,\pi]$ is deformed down so that the new contour, denote it $\cL'$, is below $i\mup/a$.
 Then the residue terms are the same as in the typical case for $x'_k\ge 0$ above but the choice of the parameters of the sinh-acceleration
 becomes more involved. We choose $\om_1\in \bR, b>0, \om>0, d>0$ and the new $\gam\in (0, \gap)$ so that, for
 all $d'\in (0,d)$, the deformed strip $\chi_{\om_1,b,\om}(S_{(-d,d)})$ is above the poles $\xi^+_{\pm 1}$ but below $i\mup/a$.
 
 We fix a moderately large $\mu_a$, e.g., $\mu_a=3$, and set $\mup'=\min\{\mup,\mu_a\}/a$, $\mum'=\min\{\ln(2+\sqrt{3}),\mup'/2\}$,
 $\gam'=\arctan((\ln(2+\sqrt{3})-\mum')/\pi)$, $\gap'=\gap$, and then choose the remaining parameters following the general scheme
 in \cite{BoyarLevenSinh19} with $\ga'_\pm, \mu_\pm'$ in place of $\ga_\pm, \mu_\pm$.

  \subsection{The case $x'_k \le 0$}
  If $\mum/a<-\ln(2+\sqrt{3})$,
 the integrand on the RHS of \eq{Intcrossm} is analytic in $i(\mum/a, -\ln(2+\sqrt{3}))+\cC_{\gam,0}$, and decays at infinity as
 the integrand in the pricing formula for the OTM European call \cite{BoyarLevenSinh19}
 in KoBoL model,
 with $(\mum/a,-\ln(2+\sqrt{3}))$ in place of $(\lm,-1)$; the cone of analyticity where $e^{-\barDe\psi^0(a\xi)}$ decays is
 $\cC_{\gam, 0}$, where $\gam=-\min\{\pi/2, \pi/(2\nu)\}$. Hence, we deform the contour of integration down using
 $\om=\gam/2$, $d_0=-\om$, and choose the remaining parameters using the general scheme as well.

 If $\mum/a\ge -\ln(2+\sqrt{3})$ or 
 $\mum/a$ very close to $-\ln(2+\sqrt{3})$, we take $\om_0<-\ln(2+\sqrt{3})$, and  deform the line of integration $\bR$ into
 the line $\{\Im\xi=\om_0\}$, whose part below $[-\pi,\pi]$ is deformed up so that the new contour, denote it $\cL'$, is above $i\mum/a$.
 Then the residue terms are the same as in the typical case for $x'_k\le 0$ above but the choice of the parameters of the sinh-acceleration
 becomes more involved. We choose $\om_1\in \bR, b>0, \om<0, d>0$ and the new $\gap\in (\gam,0)$ so that, for
 all $d'\in (0,d)$, the deformed strip $\chi_{\om_1,b,\om}(S_{(-d,d)})$ is below the poles $\xi^-_{\pm 1}$ but above $i\mum/a$.
 
 We fix a moderately large $\mu_a$, e.g., $\mu_a=3$, and set $\mum'=\max\{\mum,-\mu_a\}/a$, 
 $\mup'=\max\{-(2+\sqrt{3}), \mum'/2\}$,
 $\gap'=-\arctan((\ln(2+\sqrt(3)+\mup')/\pi)$, $\gam'=\gam$, and then choose the remaining parameters following the general scheme
 in \cite{BoyarLevenSinh19} with $\ga'_\pm, \mu_\pm'$ in place of $\ga_\pm, \mu_\pm$.

 %%%%%%%%%%%%%%%%%%%%%%%%%%%%%%%%%
 \subsection{Formulas for the coefficients in B-spline projection method, quadratic splines}\label{coeffquadratic}
%%%%%%%%%%%%%%%%%%%%%%%%%%%%%%%%%
 In this case, $\htphi(\xi)$ is given by
 \bbe\label{htphiquadr}
 \htphi(\xi)=\frac{480(\sin(\xi/2)/\xi)^3}{33+26\cos\xi+\cos(2\xi)}.
 \ee
The differences with the case of linear splines are:
 \begin{enumerate}[(a)]
 \item
 there are four series of simple poles, not two;
 \item
 the series are of the form $\xi^\pm_{j;\ell}=(2\ell+1)\pi\pm i v_j$,   where $v_j>0$, $j=1,2$;
 \item
 the calculation of the residues is messier than in the linear case but straightforward, nevertheless;
 \item
 instead of the sums of two infinite series, we need to calculate the sums of four series;
 \item
 $v_2$ is rather large, hence, in wide regions of the parameter space, the deformations of the contour
 of integration at the first step is not flat, and, after all poles are crossed, the sinh-acceleration can be applied only with
 smaller $\gap-\gam$ than in the typical case for linear splines. If $\gap-\gam$ is too small,
 the alternative is not to cross the poles which are the closest ones
 to the imaginary axis. 
 \end{enumerate}
 The series enjoy the same properties as in the linear case. 
 The poles are found in 2 steps. First, we represent the denominator in \eq{htphiquadr} as the quadratic polynomial 
 $
 16+13\cos \xi+\cos^2\xi
 $
 in $\cos \xi$ and find its roots
 $
 \cos \xi=-0.5(13\pm \sqrt{105}).
 $
 Then we set $a_j=13+(-1)^j \sqrt{105}=0$, solve the equations $w+w^{-1}+a_j=0$, $j=1,2$,  and find
 \[
 w^\pm_j=-0.5(a_j\pm\sqrt{a_j^2-4}).
 \]
 We have $w^+_jw^-_j=1$, $-w^+_j>1>-w^-_j>0$ Hence, we have the representations
 (b) with $v_j=\ln(-w^+_j)>0$, $j=1,2$.
 
 \subsection{Formulas for the coefficients in B-spline projection method, cubic splines}\label{coeffcubic}
 As in the case of linear and cubic splines, the numerator in the formula for $\htphi(\xi)$ is an entire function;
 the denominator  is 
 \[
 D(\cos(\xi))=1208+1191\cos\xi+120\cos(2\xi)+\cos(3\xi)=4(297+297\cos\xi+60\cos^2\xi+\cos^3\xi).
 \]
 We have $D(-1)>0, D((-2)<0, D(-20)>0$, and $D(- \infty)=-\infty$. Hence, the polynomial $D(z)$ has three 
 roots on $(-\infty, -1)$, which we denote  $-0.5a_j, j=1,2,3,$. Solving 
  the equations $w+w^{-1}+a_j=0$, we find  three series of poles $\xi^\pm_{j;\ell}=(2\ell+1)\pi\pm i v_j\ln(-w_j)$, 
  $v_j=\ln(-w_j)>0$,
  $j=1,2,3$.

%%%%%%%%%%%%%%%%%%%%%%%%%%%%%%%%%
 \subsection{The General Case, Splines of higher order}\label{coeffhigher}
 %%%%%%%%%%%%%%%%%%%%%%%%%%%%%%%%%
With the Haar scaling function defined by $\varphi^{[0]}(y):= \mathbbm{1}_{[-\frac{1}{2},\frac{1}{2}]}(y)$,  the $p$-th order B-spline scaling functions are defined inductively 
\begin{equation}\label{eq: convol}
\varphi^{[p]}(x) = \varphi^{[0]}\star \varphi^{[p-1]}(x) = \int_{-\infty}^\infty \varphi^{[p-1]}(y-x)\mathbbm{1}_{[-\frac{1}{2},\frac{1}{2}]}(y)dy.
\end{equation}
 From equation \eqref{eq: convol}, the $p$-th order B-spline generator has Fourier transform
\begin{equation}\label{eq: phihat}
\widehat\varphi^{[p]}(\xi) = \left(\frac{\sin(\xi/2)}{(\xi/2)}\right)^{p+1}, \quad \widehat\varphi^{[p]}(0)=\int_{-\infty}^\infty \varphi^{[p]}(x)dx=1,
\end{equation}
which decays at a polynomial rate $L+1$ for the B-spline of order $L$.  From \eqref{eq: dualft} and an explicit representation of  $\bm \Phi(\xi)$ (see \cite{Ki14E}), we have
\begin{align*}
\widehat{\widetilde\varphi}^{[p]}(\xi)&= \left(\frac{\sin(\xi/2)}{(\xi/2)}\right)^{p+1} \mathbf \Phi^{[p]}(\xi)\\
&=\left(\frac{\sin(\xi/2)}{(\xi/2)}\right)^{p+1}\left( \int_{-\frac{p+1}{2}}^{\frac{p+1}{2}} \varphi^{[p]}(x)^2dx  + 2\sum_{k=1}^{p+1}\cos(k\xi) \int_{-\frac{p+1}{2}}^{\frac{p+1}{2}} \varphi^{[p]}(x)\varphi^{[p]}(x-k)dx \right)^{-1},
\end{align*}
where $\Phi^{[p]}(\xi)$  is a degree $p+1$ polynomial in $\cos(\xi)$.  In particular,
 we have
%\begin{align*}
%&\mathbf \Phi^{[0]}(\xi) =1 \\
%&\mathbf \Phi^{[1]}(\xi) = \frac{1}{3}(2 + \cos(\xi))\\
%&\mathbf \Phi^{[2]}(\xi) = \frac{1}{60}\left(33 + 26 \cos(\xi) + \cos(2\xi)\right)\\
%&\mathbf \Phi^{[3]}(\xi) = \frac{1}{2520}\left(1208 + 1191 \cos(\xi) + 120 \cos(2\xi) +\cos(3\xi)\right),
%\end{align*}
\[
\mathbf \Phi^{[3]}(\xi) = \frac{1}{2520}\left(1208 + 1191 \cos\xi + 120 \cos(2\xi) +\cos(3\xi)\right),
\]
and similarly for higher orders. Calculation of the series of poles is a straightforward albeit rather tedious exercise.

 %%%%%%%%%%%%%%%%%%%%%%%%%%%%%%%%%
 \section{Pricing of barrier options: explicit algorithms}\label{s:barrier pricing}
 %%%%%%%%%%%%%%%%%%%%%%%%%%%%%%%%%
 The are now several applications of the PROJ pricing method for path-dependent problems in quantitative finance, risk management, and insurance, see for example \cite{KirkNguyenCuiBermuBar17,WangZhang19,lars2019swing,kirkby2020analysis,ZhangShi2020,ZhangYonYu2020,cui2021data,kirkby2021equity}.  
Here we consider the problem of recursive valuation of an option, such as a European or barrier option.
 As it is typically the case, suppose that we maintain a value function approximation $\mathcal V(y_k,t_{m})$ over a uniformly spaced grid on $[l,u]$,
\begin{equation}
y_n = l + (n-1)\Delta, \quad n=1,\ldots, N_y,
\end{equation} 
where $\Delta := (u-l)/(N_y - 1)$.
Here $t_m$ indexes the time step in our recursive procedure, $m=0,1, \ldots, M$, where $t_m:=\bar \Delta m$.  Similarly, define a transition density grid 
\begin{equation}
x_k = x_1 + (k-1)\Delta, \quad k=1,\ldots, N_x,
\end{equation}
where $N_x:=2N_y$. The choice of $x_1$ and $N_x$ is discussed in  Appendix \ref{sect:AdaptiveErr}.
Using the SINH procedure detailed in Section \ref{s:SINH-PROJ}, we calculate the array of projection coefficients $\bar{\bm \beta}_a=(\bar \beta_{a,k})_{k=1}^{N_x}$ corresponding to $(x_k)_{k=1}^{N_x}$, which yields the projected density approximation
\begin{equation}
\bar p_{\bar \Delta}(x):=\sum_{k=1}^{N_x}\bar \beta_{a,k}\varphi_{a,k}(x).
\end{equation} 
In this section, we assume that $\varphi\equiv\varphi^{[1]}$ is the linear basis generator.
Moving backward in time from $t_m$ to $t_{m-1}$, we  compute the convolution of the transition density with the value function to obtain
\begin{align} \label{eq: val2}
\mathcal V(y_n,t_{m-1}) &=e^{-r\bar \Delta} \int_l^u \mathcal V(y,t_m)\bar p_{\bar \Delta}(y-y_n)dy \nonumber \\ 
& = \Upsilon^{\bar \Delta}_{a}\sum_{k=1}^{N_y}\bar \beta_{N_y+(k-n)} \cdot a^{1/2} \int_{l}^{u} \mathcal V(y,t_m)  \cdot a^{1/2}\varphi(a(y-y_k))dy
\end{align}
for $n=1,..., N_y,$ where $\Upsilon^{\bar \Delta}_{a}:=e^{-r\bar \Delta}a^{-1/2}$. To simplify notation, define the vector $\bm{\theta}_m=(\theta_{m,k})_{k=1}^{N_y}$ of \emph{value coefficients}
\begin{equation}\label{eq:thetamk}
\theta_{m,k} =  a^{1/2}\int_{l}^{u} \mathcal V(y,t_m) \cdot a^{1/2}\varphi(a(y-y_k))dy, \quad k = 1,\ldots, N_y.
\end{equation}
Thus, at each iteration we have the value update formula
\[
\mathcal V(y_n,t_{m-1}) :=\Upsilon^{\bar \Delta}_{a}\sum_{k=1}^{N_y}\bar \beta_{N_y+(k-n)} \cdot \theta_{m,k}, \quad n=1, \ldots, N_y,
\]
which can be written compactly as 
%\[
%\bm{\mathcal V}_{m-1}= \bm T \bm{\theta}_m, \quad m=M,\ldots, 1,
%\]
\[
\bm{\mathcal V}_{m-1}=\mathrm{conv}(\bar{\bm \beta}_a,\bm{\theta}_{m}), \quad m=M,\ldots, 1.
\]
Here we define $\bm{\mathcal V}_{m-1}=(\mathcal V_{m-1,n})_{n=1}^{ N_y}$, where $\mathcal V_{m-1,n} :=\mathcal V(y_n,t_{m-1})$, and $\mathrm{conv}$ is the discrete convolution operator. One can use
the built-in MATLAB function but the explicit efficient realization in terms of FFT and iFFT 
(see \ref{sss:convolution-MATLAB}) 
is usually more efficient.

Thus, it remains only to compute the numerical coefficients $\bm{\theta}_m$. In the first stage of the recursion, at $t_M$, $\mathcal V(y,t_M)=G(y)$ is given explicitly by the terminal payoff, 
\begin{equation}\label{eq:thetaMk}
\theta_{M,k} =  a^{1/2}\int_{l}^{u}  G(y) \cdot a^{1/2}\varphi(a(y-y_k))dy, \quad k = 1,\ldots, N_y,
\end{equation}
where $G(y)$ is the payoff function. For typical contracts, $\theta_{m,k}$ can be computed analytically with ease, as in \cite{Ki14C}.\footnote{See Section 3.8-3.9 of \cite{Ki14C} for down-and-out options, and Section 3.10 for up-and-out options.} For convenience, we provide explicit formulas for $\theta_{M,k}$ in Appendix \ref{sect:CoeffBarr} for common varieties of barrier contracts.

For $m<M$, we can approximate the value function accurately using a polynomial approximation, and compute the coefficients again in closed form. Hence, for each interior interval $I_k:=[y_{k-1},y_{k+1}]$, $k=2,..., N_y -1$, we define the local quadratic interpolation
%\footnote{This step is similar to an intermediate computation in \cite{DeLe14}. }
% They use a B-spline interpolation over the full grid to approximate the Fourier transformed value.  In contrast, our %approach uses a local Lagrange interpolation applied separately to overlapping subintervals.  This allows us to obtain the explicit form of $\theta_{m,k}$.} of $\mathcal V(y,t_m)$ on $I_k$ by
\begin{align}\label{eq: quadint}
\widetilde {\mathcal V}_{m,k}(y) &= \mathcal V_{m,k-1} \frac{(y-y_k)(y-y_{k+1})}{2\Delta^2} - \mathcal V_{m,k}\frac{(y-y_{k-1})(y-y_{k+1})}{\Delta^2} \nonumber\\
&\qquad +\mathcal V_{m,k+1}\frac{(y-y_{k-1})(y-y_{k})}{2\Delta^2}.
\end{align}
This allows us to obtain the explicit form of $\theta_{m,k}$. (In \cite{DeLe14}, 
a different parametrization of piece-wise polynomials
was used.)
Since $\varphi(a(y-y_k))\widetilde {\mathcal V}_{m,k}(y) $ is piecewise cubic on $I_k$, by splitting the interval in half, integration by Simpson's rule on each subinterval is exact:
\begin{equation}\label{thetttt}
a^{1/2}\int_{I_k}\widetilde {\mathcal V}_{m,k}(y)  \cdot a^{1/2}\varphi(a(y-y_k))dy = [ \widetilde {\mathcal V}_{m,k-1/2}+ \widetilde {\mathcal V}_{m,k} + \widetilde {\mathcal V}_{m,k+1/2}]/3,
\end{equation}
where $\widetilde {\mathcal V}_{m,k-1/2}:=\widetilde {\mathcal V}_{m,k}(y_k - \Delta/2)$, $\widetilde {\mathcal V}_{m,k+1/2}:=\widetilde {\mathcal V}_{m,k}(y_k + \Delta/2)$, and $\widetilde {\mathcal V}_{m,k} =  {\mathcal V}_{m,k}.$  From equation \eqref{eq: quadint}, 
\begin{align*}
\widetilde {\mathcal V}_{m,k-1/2} &= \frac{3}{8}\mathcal V_{m,k-1} + \frac{3}{4}\mathcal V_{m,k} - \frac{1}{8}\mathcal V_{m,k+1}\\
\widetilde {\mathcal V}_{m,k+1/2} & = -\frac{1}{8}\mathcal V_{m,k-1} + \frac{3}{4}\mathcal V_{m,k} + \frac{3}{8}\mathcal V_{m,k+1}.
\end{align*}
To handle the boundary intervals $I_1:=[y_1,y_2]$ and $I_{ K}:= [y_{ N_y-1},y_{ N_y}]$, we approximate using a cubic interpolating polynomial.  Plugging in these expressions into \eqref{thetttt} gives us the update rule for the coefficients:
\begin{equation} \label{eq: coeffVm}
\theta_{m,k} :=  \begin{cases}\left[13\mathcal V_{m,1} + 15 \mathcal V_{m,2} - 5\mathcal V_{m,3} + \mathcal V_{m,4}  \right]/48& \mbox{ }k = 1 \\
\left[\mathcal V_{m,k-1} + 10 \mathcal V_{m,k}+ \mathcal V_{m,k+1}\right]/12& \mbox{ } k= 2,..., N_y-1 \\
\left[13\mathcal V_{m, k} + 15 \mathcal V_{m,k-1} - 5\mathcal V_{m, k-2} + \mathcal V_{m, k-3}  \right]/48& \mbox{ } k = N_y\end{cases}
\end{equation}
Thus, the computational effort to compute ${\bm \theta_m}$ is $\mathcal O(N_y)$.

We summarize the steps in this procedure as follows:
\begin{enumerate}
\item Initialize the grid as in Appendix \ref{sect:AdaptiveErr}.
\item Calculate the array $\bar{\bm \beta}=(\bar \beta_{a,k})_{k=1}^{N_x}$ using the algorithm in Sect. \ref{s:SINH-PROJ}.
\item Initialize $\bm \theta_M$ analytically from terminal payoff, by integrating \eqref{eq:thetaMk} (see Appendix \ref{sect:CoeffBarr})
\item For $m=M,M-1 \ldots, 1$, 
\begin{enumerate}
\item Calculate $\bm{\mathcal V}_{m-1}=\mathrm{conv}(\bar{\bm \beta}_a,\bm{\theta}_{m})$
\item Calculate $\bm{\theta}_{m-1}$ using \eqref{eq: coeffVm}
\end{enumerate}
\item Calculate $\mathcal V(y,t_{0}) $  for the points of interest using the interpolation applied to $\bm{\mathcal V}_{0}$
\end{enumerate}
 %%%%%%%%%%%%%%%%%%%%%%%%%%%%%%%%%
 \section{Numerical examples}\label{s:numer}
 %%%%%%%%%%%%%%%%%%%%%%%%%%%%%%%%%
 This section provides a series of numerical experiments to demonstrate the computational advantages of the SINH acceleration method.  All experiments are conducted in MATLAB 2017 on a personal computer with Intel(R) Core(TM) i7-6700 CPU @3.40GHz.
 
 Throughout, we will consider the following  L\'evy example, which belongs to the CGMY subclass of the KoBoL family, with characteristic exponent $\psi(\xi)=-i\mu\xi + \psi_0(\xi)$, where
 \begin{equation}
\psi_0(\xi)=c\Gamma(-\nu)[\lp^\nu-(\lp+i\xi)^\nu+(-\lm)^\nu-(-\lm-i\xi)^\nu].
\end{equation}
 We specify the parameters
\begin{align}\label{eq:params}
&\text{Test I: } \quad \nu = 1.2, \quad \lambda_+ = 11, \quad \lambda_- = -4, \quad m_2=0.1,\\
&\text{Test II:} \quad \nu = 0.3, \quad \lambda_+ = 8, \quad \lambda_- = -9, \quad \text{ }m_2=0.1,\label{eq:params2}
\end{align}
where $m_2=\psi^{\prime\prime}(0)$ is the second instantaneous moment of the process, and define 
\begin{equation}
 c = \frac{m_2}{\Gamma(2-\nu)}\frac{1}{(-\lambda_-)^{\nu - 2} + \lambda_+^{\nu - 2}}.
 \end{equation}
We define the interest rate $r=0.02$ and dividend yield $q=0$.  The drift is determined by the martingale condition, $\mu=r - q - \psi_0(-i)$.

\begin{figure}[h!t!b!]
\centering     %%% not \center
\includegraphics[width=.7\textwidth]{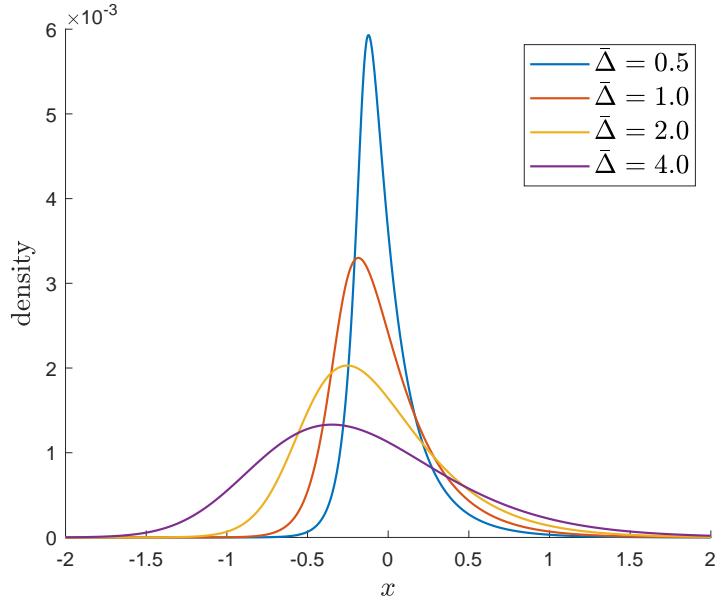}
\caption{KoBoL estimated densities for several values of $\bar \Delta$ and Test I parameters in \eq{eq:params}.}\label{fig:KobolDensities}
\end{figure}

  %%%%%%%%%%%%%%
 \subsection{Density Estimation and FFT Aliasing}
  %%%%%%%%%%%%%%
We start by demonstrating a primary deficiency of the FFT when estimating the transition density, which is resolved by using the SINH method to compute the density coefficients.  The issue is related to truncation of the density support. For the SINH method, this truncation results in just one source of error, namely the loss of mass 
of the density, outside of the truncation interval. The coefficients are not otherwise impacted. By contrast, the FFT is subject to an additional source of error, in that the truncation of the density can actually reduce the accuracy of the computed coefficients near the boundaries the truncated support, as we illustrate in more detail below.

Figure \ref{fig:KobolDensities} shows the KoBoL transition densities  (plotted in log-space) with parameters in \eq{eq:params}, computed using SINH for several values of the time step-size $\bar \Delta$.  The densities exhibit an especially heavy right tail, due to $\lambda_+ > |\lambda_-|$, with marked asymmetry.  When using the FFT to compute coefficients (or other related computations), the well known \emph{aliasing} effect emerges
because, in effect, FFT replaces a given function by a periodic one, which reduces the accuracy of the computed coefficients.
In particular, it is the slow tail decay, coupled with tail asymmetry, that causes difficulties when using the FFT.  The left panel of Figure \ref{fig:Aliasing} demonstrates this problem. Since the left tail is much heavier than the right, the implied periodicity causes a loss of accuracy for coefficients near the left boundary. 

 \begin{figure}[h!t!b!]
\centering     %%% not \center
\subfigure{\includegraphics[width=.52\textwidth]{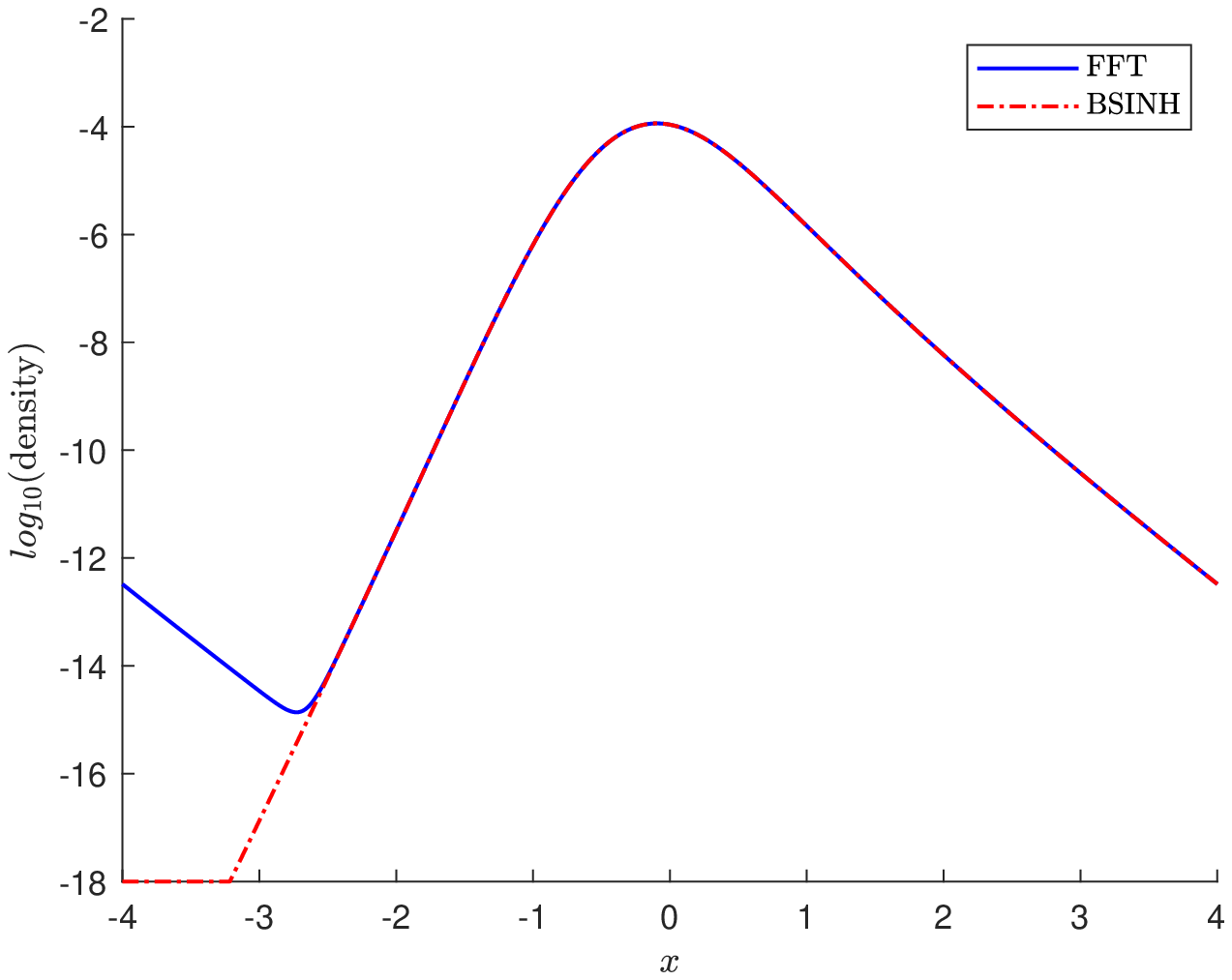}}\hspace{-2em}
\subfigure{\includegraphics[width=.52\textwidth]{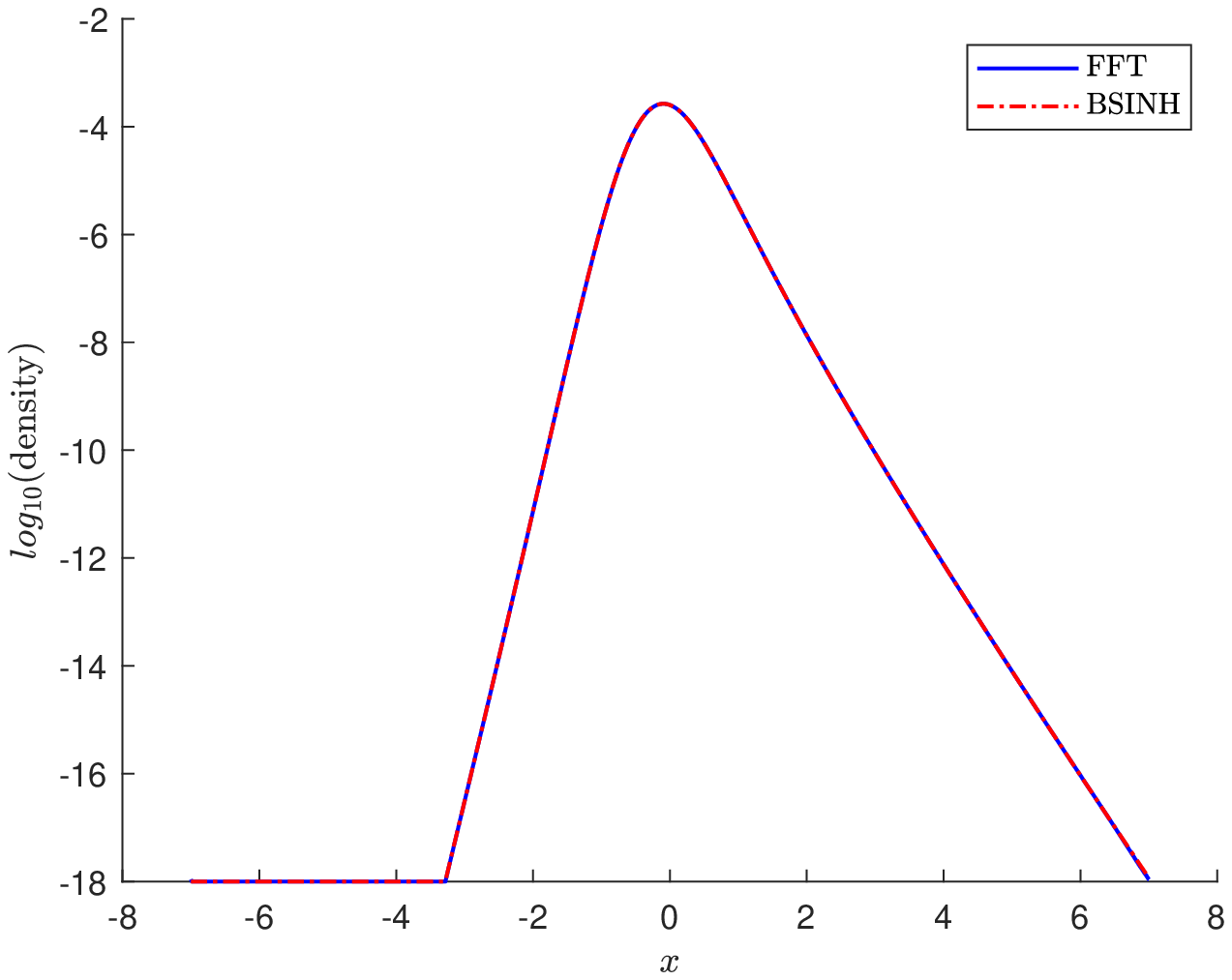}}
\caption{Visualization of aliasing effect in terms of log density using Test I parameters in \eq{eq:params}, and $\bar \Delta=1$. Left: Density is truncated to the support $[-4, 4]$, which fails to cover the mass in right tail. The heavy right tail causes aliasing of FFT on the left tail, due to assumed periodicity of FFT. Right: Removal of aliasing by increasing support to $[-7, 7]$ to cover heavy right tail.}\label{fig:Aliasing}
\end{figure}

Naturally, the aliasing issue can be resolved by expanding the truncated density support, in this case from $[-4,4]$ to $[-7,7]$ (after experimentation), as illustrated in the right panel of Figure \ref{fig:Aliasing}. There are two issues with this. The first is that we don't know \emph{in advance} that $[-7, 7]$ is sufficient, and while it can be determined automatically by detecting aliasing effects after coefficient computation, this adds additional cost and complexity to the procedure. The second issue is that even if we can automatically determine  a sufficiently wide truncation interval to avoid aliasing effects, it is wasteful of computational resources to extend the truncation grid just to avoid aliasing. Extending the grid width spreads out our $N_x$ basis elements, reducing the accuracy of the estimates (or equivalently the cost to achieve comparable accuracy).

  %%%%%%%%%%%%%%
 \subsection{European Options}
  %%%%%%%%%%%%%%
The first set of experiments demonstrated the problem with the FFT when estimating the transition coefficients. Here we show the consequences this can have on our pricing accuracy for European options.
We further demonstrate the robustness of SINH to mispecifying the truncated density support. For illustration, we consider the procedure introduced in \cite{FaOo08}, which recommends choosing the truncation half-width according to
\begin{equation}\label{eq: cossup}
\alpha: =  L_1\sqrt{|c_2|T+ \sqrt{|c_4|T}},
\end{equation}
where $c_n$ denotes the $n$-th cumulant of $\ln(S_1/S_0)$, $T$ is the contract maturity, and $L_1$ is a user-supplied parameter. While simple to implement, this procedure requires the subjective choice of the parameter $L_1$, which can have a great impact on accuracy. 

 \begin{figure}[h!t!b!]
\centering     %%% not \center
\subfigure{\includegraphics[width=.52\textwidth]{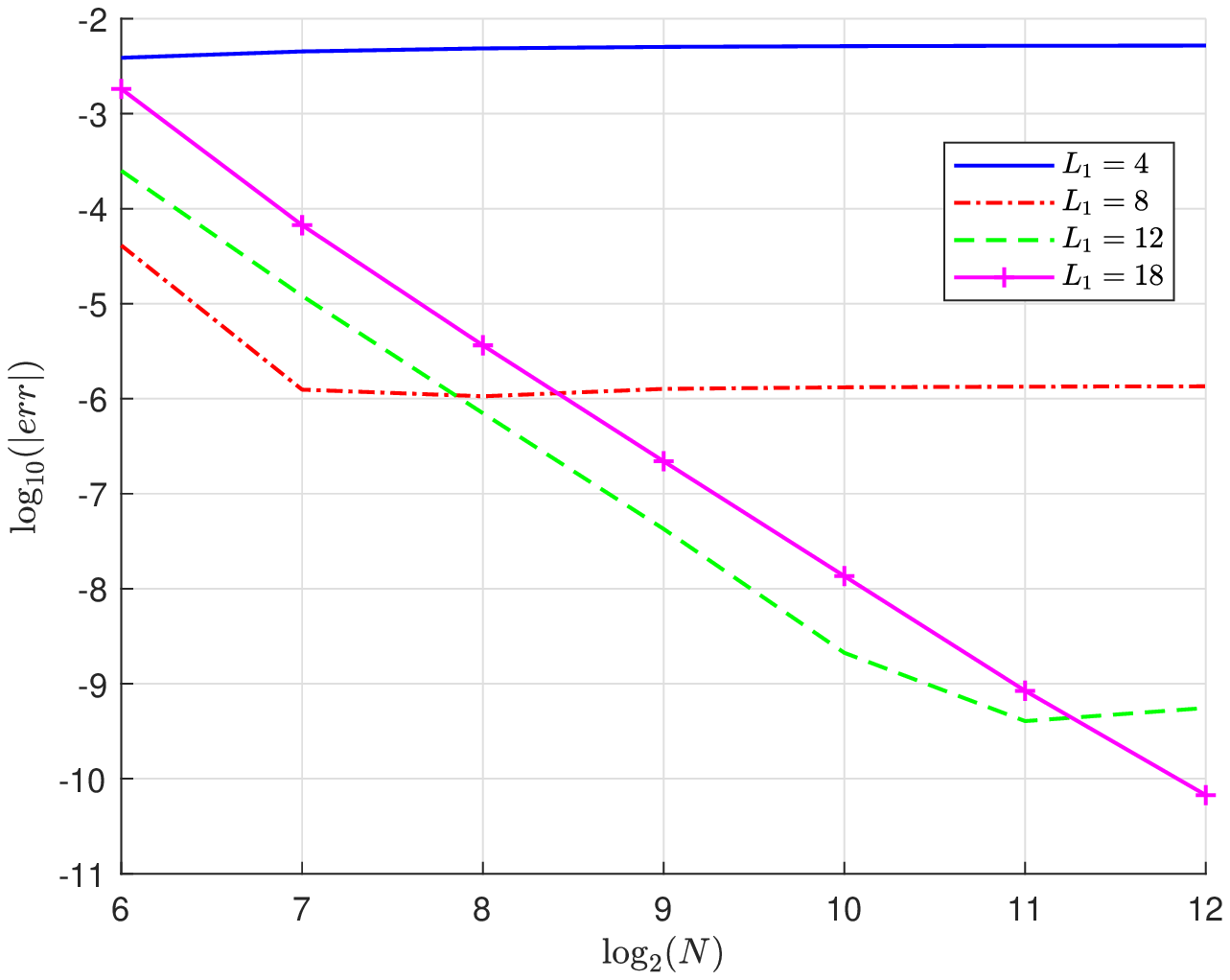}}\hspace{-2em}
\subfigure{\includegraphics[width=.52\textwidth]{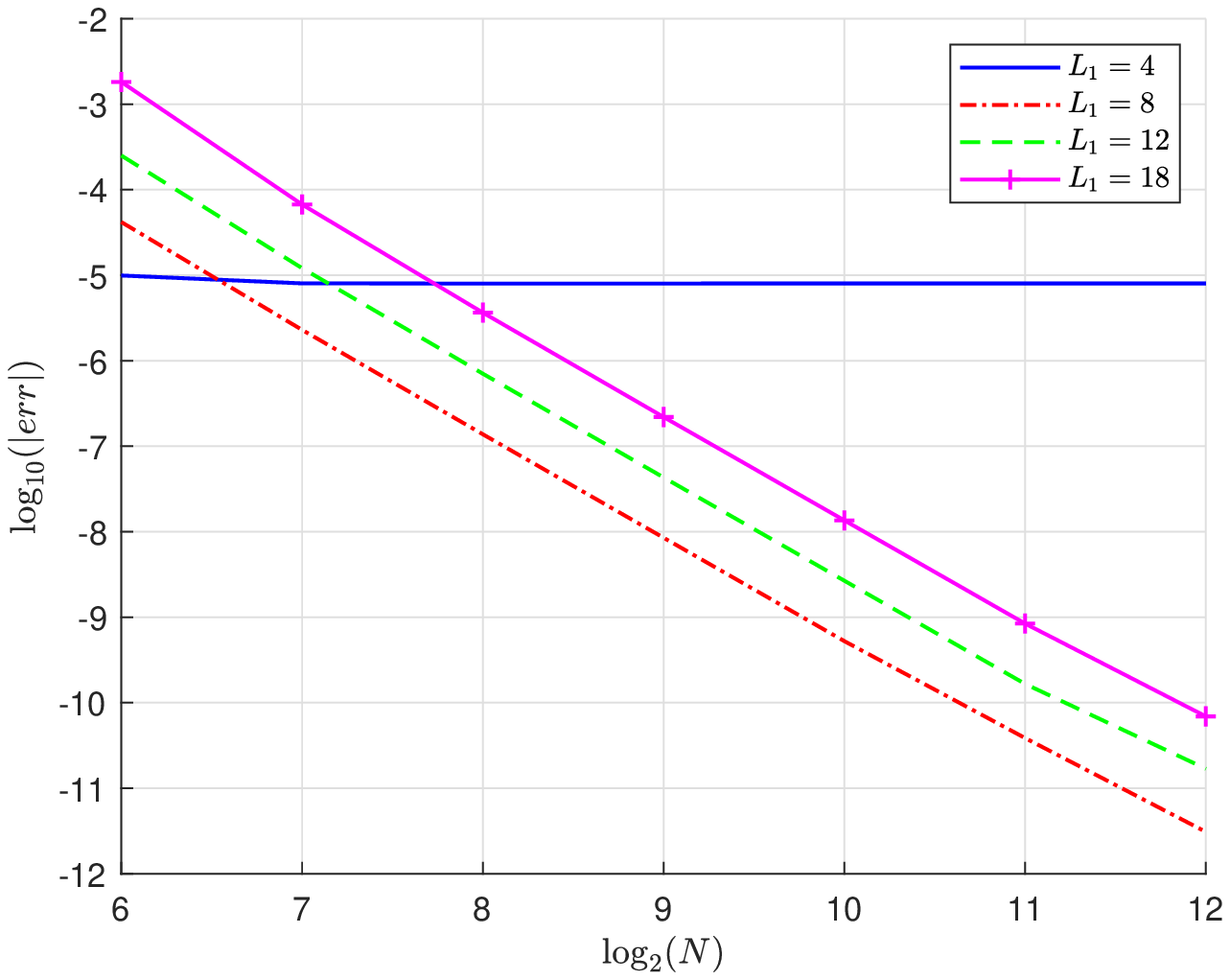}}
\caption{Robustness of SINH compared with FFT for mis-specification of truncated density support, with KoBoL Test I parameters in \ref{eq:params}. European option price convergence with linear projection coefficients determined by FFT (Left) and SINH (right).}\label{fig:EuropeanRobust}
\end{figure}
We next show that this impact can be amplified by the use of the FFT, whereas SINH is more robust. Figure \ref{fig:EuropeanRobust} demonstrates the convergence in $N$ for a European call option, with $S_0=W=100$, $T=1$. Comparing the left and right panels, which show the convergence with coefficients obtained via the FFT and SINH, respectively, we see that SINH is more accurate for small values of $L_1$. Even with $L_1=4$, SINH easily achieves beyond practical precision. Since the full density width is $2\alpha$, this suggests that we can use many fewer basis elements with SINH to achieve the same level of accuracy.
% In particular, the FFT requires a larger choice of $\alpha$, and $\Delta =2 \alpha / (N-1)$. Hence, doubling $\alpha$ reduces the grid resolution by a factor of 4.

 \begin{table}[h!t!b!]
\centering
\scalebox{0.88}{
\begin{tabular}{c|cc|cc|cc|cc}
%\hline
&\multicolumn{4}{c|}{\underline{ $L_1=4$} } &\multicolumn{4}{c}{ \underline{$L_1=8$} } \\
&\multicolumn{2}{c}{ FFT } & \multicolumn{2}{c|}{SINH}&\multicolumn{2}{c}{FFT}& \multicolumn{2}{c}{SINH}\\ %[-1ex]
%&&&&&&\\ [-1ex]
$log_2(N_y)$ & Price & $|\text{Error}|$ & Price & $|\text{Error}|$ & Price & $|\text{Error}|$ & Price & $|\text{Error}|$ \\\hline
&&&&&&&&\\ [-1.5ex]
5 & 0.83166761 & 5.82e-04 & 0.83163089 & 5.45e-04 & 0.83386326 & 2.78e-03 & 0.83390363 & 2.82e-03 \\ 
6 & 0.83116978 & 8.40e-05 & 0.83115683 & 7.10e-05 & 0.83163074 & 5.45e-04 & 0.83163072 & 5.45e-04 \\ 
7 & 0.83110239 & 1.66e-05 & 0.83108898 & 3.18e-06 & 0.83115681 & 7.10e-05 & 0.83115681 & 7.10e-05 \\ 
8 & 0.83109959 & 1.38e-05 & 0.83108594 & 1.45e-07 & 0.83108896 & 3.16e-06 & 0.83108895 & 3.16e-06 \\ 
9 & 0.83109758 & 1.18e-05 & 0.83108582 & 2.42e-08 & 0.83108592 & 1.25e-07 & 0.83108591 & 1.21e-07 \\ 
10 & 0.83109762 & 1.18e-05 & 0.83108581 & 1.90e-08 & 0.83108580 & 8.90e-09 & 0.83108580 & 5.75e-09 \\
[-1.5ex]
&&&&&&&&\\ \hline
\end{tabular}}
\caption{\small{UOC barrier option converge, with $K=S_0=100$, $M=12$, $T=1.0$, and upper barrier $U=120$.  KoBoL Test I  parameters in \eq{eq:params}.
%Comparison for two truncation levels defined by $L_1=4$ and $L_1=8$, for FFT vs SINH.
}}\label{UOC}
\end{table}

   %%%%%%%%%%%%%%
 \subsection{Single Barrier Options}
  %%%%%%%%%%%%%%
 We next consider the case of barrier options, to demonstrate that the SINH approach can also be used to improve the computational performance of path-dependent pricing under the PROJ method.  The first example we consider is that of an up-and-out call (UOC) option, with $T=1$ year to maturity, $M=12$ monitoring dates (monthly), struck at-the-money ($K=S_0=100$), with an upper barrier level at $U=120$.  For convenience, we continue to use formula \eqref{eq: cossup} to determine the truncation width of the density, which is parameterized by $L_1$.  This allows us to directly control the truncation width, and demonstrate the loss in accuracy that occurs if the truncation width is chosen too small.

   \begin{table}[h!t!b!]
\centering
\scalebox{0.88}{
\begin{tabular}{c|cc|cc|cc|cc}
%\hline
&\multicolumn{4}{c|}{\underline{ $L_1=4$} } &\multicolumn{4}{c}{ \underline{$L_1=8$} } \\
&\multicolumn{2}{c}{ FFT } & \multicolumn{2}{c|}{SINH}&\multicolumn{2}{c}{FFT}& \multicolumn{2}{c}{SINH}\\ %[-1ex]
%&&&&&&\\ [-1ex]
$log_2(N_y)$ & Price & $|\text{Error}|$ & Price & $|\text{Error}|$ & Price & $|\text{Error}|$ & Price & $|\text{Error}|$ \\\hline
&&&&&&&&\\ [-1.5ex]
5 & 2.52550469 & 1.40e-02 & 2.52544628 & 1.39e-02 & 2.82699044 & 3.15e-01 & 2.79932530 & 2.88e-01 \\ 
6 & 2.51273033 & 1.19e-03 & 2.51271077 & 1.17e-03 & 2.52548647 & 1.39e-02 & 2.52544624 & 1.39e-02 \\ 
7 & 2.51158699 & 4.32e-05 & 2.51156667 & 2.29e-05 & 2.51271072 & 1.17e-03 & 2.51271073 & 1.17e-03 \\ 
8 & 2.51157097 & 2.72e-05 & 2.51154362 & 1.20e-07 & 2.51156663 & 2.29e-05 & 2.51156663 & 2.29e-05 \\ 
9 & 2.51157536 & 3.16e-05 & 2.51154382 & 8.44e-08 & 2.51154356 & 1.82e-07 & 2.51154355 & 1.88e-07 \\ 
10 & 2.51157550 & 3.18e-05 & 2.51154383 & 8.93e-08 & 2.51154375 & 6.65e-09 & 2.51154374 & 1.65e-09 \\ 
[-1.5ex]
&&&&&&&&\\ \hline
\end{tabular}}
\caption{\small{DOP  barrier option  convergence, with $K=105, S_0=100$, $M=24$, $T=1.0$, and lower barrier $L=80$.   KoBoL Test I  parameters in \eq{eq:params}.
%Comparison for two truncation levels defined by $L_1=4$ and $L_1=8$, for FFT vs SINH.
}}\label{DOP}
\end{table}
 
 Table \ref{UOC} illustrates the convergence of FFT and SINH in the case of $L_1=4$ and $L_1=8$, as a function of the number of value grid points $N_y$.  When $L_1=4$, we have clearly underspecfied the truncation width, and the FFT method ceases to improve in accuracy beyond $3e-05$, while SINH reaches $9e-08$ for the same width.  As for the European case, by increasing $L_1$ we can achieve comparable accuracy with the FFT method, but a greater computational cost. In particular, doubling $L_1$ doubles the overall algorithm cost.  We perform another experiment for a down-and-out put (DOP) option in Table \ref{DOP} with Test I parameters, and in Table \ref{DOP2} with Test II parameters, where we observe similar results to the UOC example.

  \begin{table}[h!t!b!]
\centering
\scalebox{0.88}{
\begin{tabular}{c|cc|cc|cc|cc}
%\hline
&\multicolumn{4}{c|}{\underline{ $L_1=4$} } &\multicolumn{4}{c}{ \underline{$L_1=8$} } \\
&\multicolumn{2}{c}{ FFT } & \multicolumn{2}{c|}{SINH}&\multicolumn{2}{c}{FFT}& \multicolumn{2}{c}{SINH}\\ %[-1ex]
%&&&&&&\\ [-1ex]
$log_2(N_y)$ & Price & $|\text{Error}|$ & Price & $|\text{Error}|$ & Price & $|\text{Error}|$ & Price & $|\text{Error}|$ \\\hline
&&&&&&&&\\ [-1.5ex]
5 & 2.79848853 & 1.46e-04 & 2.79809624 & 2.47e-04 & 2.69181661 & 1.07e-01 & 2.70127090 & 9.71e-02 \\ 
6 & 2.79825515 & 8.78e-05 & 2.79824741 & 9.55e-05 & 2.79848275 & 1.40e-04 & 2.79809618 & 2.47e-04 \\ 
7 & 2.79833286 & 1.01e-05 & 2.79832651 & 1.64e-05 & 2.79824967 & 9.33e-05 & 2.79824740 & 9.55e-05 \\ 
8 & 2.79835392 & 1.10e-05 & 2.79834147 & 1.47e-06 & 2.79832311 & 1.98e-05 & 2.79832651 & 1.64e-05 \\ 
9 & 2.79835750 & 1.46e-05 & 2.79834294 & 2.16e-09 & 2.79834115 & 1.79e-06 & 2.79834147 & 1.47e-06 \\ 
10 & 2.79835747 & 1.45e-05 & 2.79834296 & 1.47e-08 & 2.79834295 & 1.88e-09 & 2.79834294 & 6.26e-10 \\
[-1.5ex]
&&&&&&&&\\ \hline
\end{tabular}}
\caption{\small{DOP barrier option converge, with $K=100, S_0=100$, $M=6$, $T=0.5$, and lower barrier $L=80$.   KoBoL Test II  parameters in \eq{eq:params2}.
%Comparison for two truncation levels defined by $L_1=4$ and $L_1=8$, for FFT vs SINH.
}}\label{DOP2}
\end{table}

   %%%%%%%%%%%%%%
 \subsection{Double Barrier Options}
  %%%%%%%%%%%%%%
 Finally, we consider the case of a double barrier option.  Unlike a single barrier (which has a semi-infinite continuation domain), the double barrier option price depends only on the density coefficients that overlap the compact interval $[L,U]$. Hence, only those coefficients should be computed before performing the convolution steps.  For SINH, this is not a problem, while for FFT it leads to serious loss of accuracy due to aliasing, particularly when $[L,U]$ is narrow.  The remedy proposed in Algorithm 4 of  \cite{Ki14C} is to compute the coefficients by FFT over an enlarged interval, and restrict the convolution to those coefficients over $[L,U]$.  This increases the cost of initialization, while the remaining computational cost for the convolution steps is unchanged. We refer to this anti-aliasing FFT approach as FFT-AA.
 Note that FFT-AA is a simplified version of the refined FFT technique
 developed and used in \cite{single,BLdouble} in  Carr's
 randomization algorithm (essentially, method of lines) applied to price barrier options with continuous monitoring.
 
    \begin{table}[h!t!b!]
\centering
\scalebox{0.9}{
\begin{tabular}{c|ccc|cc|ccc}
%\hline
&\multicolumn{3}{c}{ SINH } & \multicolumn{2}{c}{FFT}&\multicolumn{3}{c}{FFT-AA}\\ %[-1ex]
%&&&&&&\\ [-1ex]
$log_2(N_y)$ & Price & $|\text{Error}|$ & Rate & Price & $|\text{Error}|$ & Price & $|\text{Error}|$ & NX\\\hline
&&&&&&&&\\ [-1.5ex]
3 & 0.68580761 & 1.27e-03 & -- & 0.68804791 & 3.51e-03 & 0.68580776 & 1.27e-03 &6 \\ 
4 & 0.68462085 & 8.05e-05 & 3.98 & 0.68811821 & 3.58e-03 & 0.68462085 & 8.05e-05 &6\\ 
5 & 0.68454314 & 2.83e-06 & 4.83 & 0.68886527 & 4.32e-03 & 0.68454314 & 2.83e-06 &6\\ 
6 & 0.68454041 & 1.02e-07 & 4.80 & 0.68933390 & 4.79e-03 & 0.68454041 & 1.02e-07 &6\\ 
7 & 0.68454031 & 4.25e-09 & 4.58 & 0.68958569 & 5.05e-03 & 0.68454031 & 4.25e-09 &6\\ 
8 & 0.68454031 & 3.14e-10 & 3.76 & 0.68971588 & 5.18e-03 & 0.68454031 & 3.12e-10 &6\\
[-1.5ex]
&&&&&&&&\\ \hline
\end{tabular}}
\caption{\small{Double barrier call option convergence, with $K=S_0=100$, $M=12$, $T=1.0$, barriers $[L,U]=[80,120]$.   KoBoL Test I parameters in \eq{eq:params}.
%Comparison for two truncation levels defined by $L_1=4$ and $L_1=8$, for FFT vs SINH.
}}\label{DKO}
\end{table}
 
  Table \ref{DKO} displays the convergence results for a double barrier call option with $K=S_0=100$ and $[L,U]=[80,120]$ using Test I parameters.  First we compare SINH with ``naive FFT'', which computes the density coefficients corresponding to $[L,U]$ only. As expected  the accuracy of FFT is capped (here at $e-03$), while PROJ with SINH converges nicely. Since SINH is not subject to truncation error, we are able to see the natural convergence rate of PROJ (with coefficients computed accurately) when only the projection error is present.
   The column ``Rate'' documents the rate of convergence corresponding to the SINH prices, computed as $-\log_2(\text{Error}_k/\text{Error}_{k-1})$.  While the theoretical convergence of the linear basis is only second order, in practice we observe slightly better than \emph{fourth order} convergence on average, as shown here for the double barrier option.

      \begin{table}[h!t!b!]
\centering
\scalebox{0.9}{
\begin{tabular}{c|ccc|cc|ccc}
%\hline
&\multicolumn{3}{c}{ SINH } & \multicolumn{2}{c}{FFT}&\multicolumn{3}{c}{FFT-AA}\\ %[-1ex]
%&&&&&&\\ [-1ex]
$log_2(N_y)$ & Price & $|\text{Error}|$ & Rate & Price & $|\text{Error}|$ & Price & $|\text{Error}|$ & NX \\\hline
&&&&&&&&\\ [-1.5ex]
3 & 0.09238632 & 2.44e-04 & -- & 0.10436047 & 1.22e-02 & 0.09231996 & 1.78e-04 & 6 \\ 
4 & 0.09214933 & 6.93e-06 & 5.14 & 0.10954382 & 1.74e-02 & 0.09214752 & 5.12e-06 & 6 \\ 
5 & 0.09214327 & 8.66e-07 & 3.00 & 0.11281405 & 2.07e-02 & 0.09214339 & 9.81e-07 & 6 \\ 
6 & 0.09214264 & 2.30e-07 & 1.91 & 0.11461655 & 2.25e-02 & 0.09214264 & 2.35e-07 & 6 \\ 
7 & 0.09214243 & 2.36e-08 & 3.29 & 0.11556216 & 2.34e-02 & 0.09214243 & 2.01e-08 & 6 \\ 
8 & 0.09214241 & 9.53e-11 & 7.95 & 0.11604647 & 2.39e-02 & 0.09214241 & 2.15e-10 & 6 \\ 
[-1.5ex]
&&&&&&&&\\ \hline
\end{tabular}}
\caption{\small{Double barrier put option convergence, with $K=S_0=100$, $M=12$, $T=1.0$, barriers $[L,U]=[90,110]$.   KoBoL Test II parameters in \eq{eq:params2}.
%Comparison for two truncation levels defined by $L_1=4$ and $L_1=8$, for FFT vs SINH.
}}\label{DKO2}
\end{table}

    We also compare PROJ with the FFT-AA approach, which after extending the coefficients during initialization, is able to remove the aliasing effects, and has a similar convergence rate to SINH.  The column ``NX'' records the multiple of additional coefficients that are computed by FFT-AA compared with SINH during initialization, which is a \emph{6 fold increase} required to remove the aliasing effect. Hence, while the aliasing can be removed in this manner, it is computationally wasteful compared with SINH.   Table \ref{DKO2} performs a similar experiment, for a double barrier put options, with narrower boundaries $[L,U]=[90,110]$, and Test II parameters. Again we see the computational improvement offered by SINH in eliminating the truncation error, and reducing the computational cost.
 
  %%%%%%%%%%%%%%%%%%%%%%%%%%%%%%%%%
% \section{A short review of applications of the Fourier transform technique in finance,
% with an outline of possible applications of the main results of the paper}\label{s:other}
  %%%%%%%%%%%%%%%%%%%%%%%%%%%%%%%%%

  %%%%%%%%%%%%%%%%%%%%%%%%%%%%%%%%%
 \section{Conclusion}\label{s:concl}
  %%%%%%%%%%%%%%%%%%%%%%%%%%%%%%%%%
In the paper, we clarified relative advantages/disadvantages of different Fourier-based approaches to option pricing, 
and improved the B-spline density projection method using the sinh-acceleration technique to evaluate
the projection coefficients. This allows one  to completely separate the control of errors of the approximation of the transition operator at each time step as a convolution operator whose action is
realized using the fast discrete convolution algorithm, and errors of the calculation of the elements of the discretized pricing kernel. The latter elements are calculated one-by-one (the procedure admits an evident parallelization), without resorting to FFT technique. Instead, appropriate conformal deformations of the contour of integration in the formula for the elements, with the subsequent change of variables, are used to greatly increase the rate of convergence. The changes of variables are such that the new integrand is analytic in a strip around the line of integration, hence, the integral can be calculated with accuracy E-14 to E-15 using the simplified trapezoid rule with a moderate or even small number of points.
We explained the difficulties for accurate calculations of the coefficients using the FFT technique. Crucially, the use of the same pair of dual grids for all purposes, which is presented in many publications as a great advantage of the FFT technique, inevitably produces large errors unless extremely long grids are used. Furthermore, in some cases, even
extremely long grids may be insufficient. We illustrated the advantages of the sinh-acceleration in the application to calculation of the B-spline coefficients through a series of numerical applications.

We also explained that although the spectral filtering technique allows one to increase the speed of calculations,
spectral filters are designed to regularize the results. In applications to derivative pricing, the regularization
results in serious errors in regions of the paramount importance for risk management: near barrier and strike,
close to maturity and for long dated options, hence, ad-hoc choices of spectral filters can produce unacceptable errors. The approximations in the method of this paper also can be interpreted as spectral filters but the rigorous error bounds allow one to ensure the accuracy of the result, and the resulting error converges quickly (fourth order on average) as the approximation grid is refined.

As the first applications of the B-spline projection method with the improved procedure for the calculation of the coefficients, we considered the pricing of European as well as single and double barrier options with discrete monitoring, using backward induction in the state space. We list several sets of  error bounds used in the literature, and discussed their relative advantages/disadvantages. The numerical experiments show that the use of SINH acceleration eliminates a key source of error in the B-spline projection methodology, which originates from the FFT when computing basis coefficients. By instead using the SINH method, the results are considerably more robust to under-specifying the density truncation width, for which aliasing of the FFT causes a deterioration of accuracy due to the assume periodicity of the FFT.  While the aliasing effect can be removed with the FFT by increasing the grid width (sometimes substantially), doing so is computationally wasteful. The use of SINH acceleration allows the user to price options with a reduced truncation width at a commensurate reduction in computational effort, with added robustness to the choice of truncation width.  For illustrative purposes, we focused on the case of European and Barrier options, but the proposed methodology extends naturally to other cases for which the PROJ method has been applied, such as for problems in risk management and insurance \cite{WangZhang19,ZhangShi2020,ZhangYonYu2020,kirkby2021equity}.  The development of an improved procedure for basis coefficient calculation is beneficial across a variety of existing and future applications.

\newpage
%%%%%%%%%%%%%%%%%%%%%%%%%%%%%%%%%
\appendix
%%%%%%%%%%%%%%%%%%%%%%%%%%%%%%%%%

%%%%%%%%%%%%%
\section{Coefficient Formulas for Barrier Options}\label{sect:CoeffBarr}
%%%%%%%%%%%%%

%%%%%%%%%%%%%
\subsection{Down-and-Out Put Coefficients}
%%%%%%%%%%%%%
The value coefficients for the DOP contract, $\theta^{DOP}_{M,k}=\theta_{M,k}$, are determined in closed form for the terminal (initializing) period, using the formulas derived in  \cite{Ki14C}, to which we refer the reader for more details. Here we summarize the basic formulas, which can be combined with Algorithm \ref{GridBarrier} to price barrier options.
First recall that the terminal DOP payoff is given by
$
(K - S_0 e^y)\mathbbm{1}_{[ \ln(L/S_0), \ln(K/S_0)]}(y).
$
The nearest grid point left of $\ln(K/S_0)$ is given by 
\begin{equation}\label{eq: nbarput}
\bar n: = \lfloor(a\cdot(\ln(K/S_0) - x_1) +1 \rfloor, 
\end{equation}
and the difference and normalized difference are defined by
\begin{equation}\label{eq: rhozeta}
\rho: =\ln(K/S_0) - x_{\bar n}, \qquad \zeta: = a\cdot \rho.
\end{equation}
Let
\begin{equation}
b_3:= \sqrt{15}, \quad b_4:= \sqrt{15}/10,
\end{equation}
and
\begin{equation}\label{eq: varthetNeg10Bar}
\vartheta_{[-1,0]} \approx \frac{e^{-\Delta/2}}{18} \left( 5 \cosh(b_4 \Delta) +b_3\sinh(b_4\Delta) +4 \right),
\end{equation}
\begin{equation}\label{eq: varthet01Bar}
\vartheta_{[0,1]} \approx \frac{e^{\Delta/2}}{18} \left( 5 \cosh(b_4 \Delta) - b_3\sinh(b_4\Delta) +4 \right).
\end{equation}
Set 
\begin{equation}\label{eq: varthetStar}
\vartheta^*:= \vartheta_{[-1,0]} + \vartheta_{[0,1]}.
\end{equation}
Using the constants defined in Table \ref{ceoffadjust}, we have the formula
\begin{equation} \label{eq: coeffgnput}
\theta^{DOP}_{M,k} = \begin{dcases} \frac{K}{2}  - L\cdot \vartheta_{[0,1]}, & \mbox{ }k =1\\
K - e^{y_k}\cdot S_0\cdot (\vartheta_{[-1,0]} + \vartheta_{[0,1]}),& \mbox{ } k =2,...,\bar n - 1\\
K\left(\frac{1}{2} + \bar \delta_0^{put} - e^{-\rho}\left( \vartheta_{[-1,0]} +\delta_0^{put} \right)\right), &\mbox{ } k = \bar n \\
K\left(\bar \delta_1^{put} - e^{\Delta-\rho}\cdot\delta_1^{put}\right),& \mbox{ } N_y=\bar n +1 \\
0, & \mbox{ }k = \bar n+2,...,N_y. \end{dcases}
\end{equation}
%%%%%%%%%%%%%
\subsection{Down-and-Out Call Coefficients}
%%%%%%%%%%%%%
We next describe the formulas for DOC payoff coefficients. Define $\bar n, \rho, \zeta$ as in equations \eqref{eq: rhozeta} and \eqref{eq: nbarput}, and further
\begin{equation}
\sigma:=1- \zeta, \qquad \sigma_{\pm}:=\sigma(q_{\pm} - 1/2)= \pm \frac{\sigma}{2}\sqrt{3/5},
\end{equation}
\begin{equation}\label{eq: qplusminus}
q_{-} : = \frac{1}{2}\left(1-\sqrt{3/5}\right), \qquad q_{+} : = \frac{1}{2}\left(1+\sqrt{3/5}\right).
\end{equation}
 Coefficients corresponding to the points $x_{\bar n}$ and $x_{\bar n +1}$ are affected by the grid misalignment (i.e., not all payoff singularities/kinks align with grid points), so we introduce the following integrals:
\begin{equation}
\bar \delta_0^{call} : = \int_{1-\sigma}^1\varphi(y)dy, \qquad \delta_0^{call} : = \int_{1-\sigma}^1 \varphi(y) e^{\frac{y}{a}}dy
\end{equation}
\begin{equation}
\bar \delta_1^{call} : = \int_{-\sigma}^{0} \varphi(y)dy, \qquad \delta_1^{call} : = \int_{-\sigma}^{0} \varphi(y) e^{\frac{y}{a}}dy.
\end{equation}
Taking $\delta^{call}_0$ as an example, we apply the Gaussian quadrature approximation
\begin{align*}
\delta^{call}_0 &\approx \frac{\sigma}{18}\left [5\cdot\left( (\sigma/2 - \sigma_-)e^{\Delta \sigma_-} + (\sigma/2 - \sigma_+)e^{\Delta\sigma_+}\right) + 4 \sigma  \right]\\
&=e^{\frac{\rho +\Delta}{2}}\frac{\sigma^2}{18}\left [5\cdot\left( (1- q_-)e^{\Delta \sigma_-} + (1-q_+)e^{\Delta\sigma_+}\right) + 4\right],
\end{align*}
and similarly for $\delta^{call}_1$.  The coefficients for $\bar \delta_0^{call},\bar \delta_1^{call}$ are derived exactly.  Table \ref{ceoffadjust} collects the resulting approximation formulas.  The DOC coefficients are given by
\begin{equation} \label{eq: coeffgncall}
\theta^{DOC}_{M,k} = \begin{dcases} 0 & \mbox{ }k =1,\ldots, \bar n-1\\
K\left(\delta_0^{call} e^{-\rho} - \bar \delta_0^{call}\right)& \mbox{ } k = \bar n\\
K\left(e^{\Delta - \rho}\left(\vartheta_{[0,1]} + \delta_1^{call}\right) - \left( 1/2 + \bar \delta_1^{call}\right)\right) &\mbox{ } k = \bar n +1 \\
S_0\exp(y_k)\cdot\vartheta_* - K& \mbox{ } k=\bar n +2,..., N_y\end{dcases}
\end{equation}

\subsection{Up and Out Contracts}
For UOC options, with the exception of $\theta^{UOC}_{M,N_y}$, the coefficients are identical to equation \eqref{eq: coeffgncall}:
\begin{equation}\label{eq: coeffbarrUOC}
\theta^{UOC}_{M,k} : = \begin{dcases}
\theta^{DOC}_{M,k} & \mbox{ } k\leq N_y-1 \\
U\cdot  \vartheta_{[-1,0]}- W/2& \mbox{ } k= N_y.
\end{dcases}
\end{equation}
Similarly, for UOP options
\begin{equation} \label{eq: coeffbarrUOP}
\theta^{UOP}_{M,k} = \begin{dcases} K - \exp(y_k) S_0 \cdot \vartheta_* & \mbox{ }k =1\\
 \theta^{DOP}_{M,k}& \mbox{ }k = 2,...,N_y \end{dcases}
\end{equation}
where $\theta^{DOP}_{M,k}$ are defined in equation \eqref{eq: coeffgnput}.

\begin{table}
\centering
\scalebox{.85}{
\begin{tabular}{c|c|c}
\hline
\multicolumn{3}{c}{}  \\[-1.5ex]
 \multicolumn{3}{c}{Gaussian Quadrature Adjustment for Payoff Coefficients}\\
\hline
& &  \\[-1.5ex]
Puts & $\bar \delta_j^{put}$ & $\delta_j^{put}$\\
& &  \\[-2ex]\hline
& & \\
$j = 0$& $ \zeta \left( 1  - \frac{1}{2}\zeta\right)$&$ \frac{\zeta}{18}\left[4(2-\zeta) e^{\rho/2} + 5 \cdot \left((1 - \zeta_{-})e^{\rho_{-}} + (1 - \zeta_{+})e^{\rho_{+}} \right)  \right]$ \\
& &  \\
$j = 1$ & $ \frac{\zeta^2}{2} $  &  $ \frac{\zeta}{18} \cdot e^{-\Delta}\cdot \left[ 4 \zeta \cdot e^{\rho/2} + 5\left(\zeta_{-}e^{\rho_{-}} +\zeta_{+}e^{\rho_{+}} \right)\right] $\\
& &  \\ \hline
& &  \\[-1.5ex]
Calls & $\bar \delta_j^{call}$ & $\delta_j^{call}$\\
& &  \\[-2ex]\hline
& & \\
$j = 0$& $\frac{1}{2} + \zeta\left(\frac{\zeta}{2} - 1\right) $&$ e^{(\rho+\Delta)/2}\frac{\sigma^2}{18}\left[5\cdot\left( (1- q_-)e^{\Delta \sigma_-} + (1-q_+)e^{\Delta\sigma_+}\right) + 4\right]$ \\
& &  \\
$j = 1$ & $\sigma - \frac{1}{2}\sigma^2 $  &  $e^{(\rho-\Delta)/2} \frac{\sigma}{18}\cdot \left[ 4(\zeta +1) + 5 \left(\left(\frac{\zeta+1}{2} + \sigma_- \right)e^{\Delta\sigma_-} + \left( \frac{\zeta+1}{2} + \sigma_+ \right)e^{\Delta\sigma_+} \right)\right] $\\
& &  \\ 
\hline 
\end{tabular}}
\caption{Coefficients derived from a three point Gaussian quadrature.  }\label{ceoffadjust}
\end{table}
%%%%%%%%%%%%%%%%%%%%%
\begin{algorithm}[h!t!b!]
\small
\caption{Initialization by Automated Parameter Selection}\label{GridBarrier}
\begin{algorithmic}[1]
\State Set $L_1 \gets 10;\quad$  Initialize $N$ as in Remark \ref{rem: InitialN}%$h \gets u$ for up-and-out, $h \gets l$ for down-and-out
\State $\epsilon_1 \gets$5e-08; $\quad \epsilon_2 \gets$1e-05; $\quad \tau \gets 1.1$
\If{Down-and-out:}
\State $\alpha \gets \max\Big\{\tau \cdot (\max\{0,\ln(W/S_0)\} - l), \text{ } L_1\sqrt{c_2 T + \sqrt{c_4 T}}\Big\}$
\State $x^* \gets l + \alpha;\quad \Delta_\xi \gets \pi(N-1)/(\alpha N)$
%%%%%
\While{$|1 -F_{\Delta_\xi,N}(x^*)|\cdot \max\{S_0,W\}> \epsilon_1$}
\State $\alpha\gets \tau \alpha; \quad N\gets 2 N;\quad x^* \gets l + \alpha;\quad \Delta_\xi \gets \pi(N-1)/(\alpha N)$
\EndWhile
\State $\mathcal E \gets 2\mathcal \epsilon_2; \quad N\gets N/2$
\While{$\mathcal E > \epsilon_2$}
%\State $N\gets 2 N;\quad \Delta_\xi \gets \pi(N-1)/(\alpha N)$
%\State $\Delta \gets 2\alpha/(N-1);\quad  x_1 \gets l; \quad a \gets 1/\Delta$
\State $N\gets 2 N;\quad \Delta \gets 2\alpha/(N-1); \quad  n_0 \gets \lfloor 1 - l/ \Delta\rfloor$
\State \textbf{if} $\Delta \leq l,\text{ }$ \textbf{then} $\Delta \gets l/(1-n_0)$ \textbf{end if}
\State Use SINH to determine $\bm{\bar\beta}$.
\State $\epsilon_2 \gets |E_N - \exp((r-q)\bar \Delta)|\cdot M\quad$ ($E_N$ defined in equation \eqref{eq: EsubNBar})
\EndWhile
%%%%
\Else { Up-and-out:}
\State $\alpha \gets \max\Big\{\tau \cdot (u - \min\{0,\ln(W/S_0)\}), \text{ } L_1\sqrt{c_2 T + \sqrt{c_4 T}}\Big\}$
\State $x^* \gets u - \alpha;\quad \Delta_\xi \gets \pi(N-1)/(\alpha N)$
\While{$|F_{\Delta_\xi,N}(x^*)|\cdot \max\{S_0,W\}> \epsilon_1$}
\State $\alpha\gets \tau \alpha; \quad N\gets 2 N;\quad x^* \gets u - \alpha;\quad \Delta_\xi \gets \pi(N-1)/(\alpha N)$
\EndWhile
\State $\mathcal E \gets 2\mathcal \epsilon_2; \quad N\gets N/2$
\While{$\mathcal E > \epsilon_2$}
\State $N\gets 2N;\quad \Delta \gets 2\alpha/(N-1);\quad n_0\gets\lfloor N/2-u/ \Delta \rfloor$
\State \textbf{if} $\Delta\leq u,\text{ }$ \textbf{then} $\Delta \gets u/(N/2-n_0)$ \textbf{end if}
\State Use SINH to determine $\bm{\bar\beta}$.
\State $\epsilon_2 \gets |E_N - \exp((r-q)\bar \Delta)|\cdot M\quad$ ($E_N$ defined in equation \eqref{eq: EsubNBar})
\EndWhile
\EndIf
\State\Return: $\bm{\bar\beta},\text{ } N, \text{ } \Delta, \text{ } \alpha, \text{ } n_0$
\end{algorithmic}
\normalsize
\end{algorithm}
%%%%%%%%%%%%%%%%%%%%

\section{Adaptive algorithm for Numerical Parameter Selection}\label{sect:AdaptiveErr}
There are two main parameters to choose when applying the PROJ algorithm.  The first is $\alpha$, which specifies the truncated density support over $[-\alpha, \alpha]$.  
The location of the first grid point $x_1=-\alpha$.  In general, $\alpha$ should be chosen to capture the mass of the transition density within some tolerance. Note that for double barrier options, $\alpha$ is determined by the barriers.
The second parameter is $N_x$, which is the number of basis elements. Assuming that $\alpha$ is sufficiently large to control the truncation error, the convergence of the approximation to the price is determined by $N_x$.  

Here we recollect the adaptive parameter selection algorithm proposed in \cite{Ki14C} for pricing single barrier options.\footnote{Double barrier options are similar, but do not require the determination of $\alpha$, so only $N_x$ is updated in the procedure.}
The starting point is the selection of $\alpha$.  With $L_1 = 10$, and fixing an initial $\widetilde N$ (see Remark \ref{rem: InitialN}), we initialize $\alpha$ and $\Delta$ as $\widetilde \alpha$ and $\widetilde \Delta$, based on the contract maturity $T$:
\begin{equation}\label{eq: alphaequationbar}
\widetilde\alpha = L_1 \sqrt{T\cdot c_2 + \sqrt{T\cdot c_4}}, \qquad \widetilde \Delta = 2 \widetilde\alpha/ (\widetilde N-1)
\end{equation}
where $c_2,c_4$ are the second and fourth cumulants of $\ln(S_{t+1}/S_t)$ (see \cite{FaOo08}).  Once we have an initial choice of $\widetilde\alpha$, we update $\widetilde\alpha$ until it captures the mass of the density within a prescribed tolerance $\epsilon_1>0$, as follows.
 From \cite{FeLi08} we can represent the cumulative distribution of log return $Y_T = \ln(S_T/S_0)$ by
\begin{align}
F(x) &= \int_{-\infty}^x p_{T}(y)dy = \int_\mathbb R p_{T}(y)\mathbbm{1}_{(-\infty,x)}dy \nonumber\\
& = \mathcal F(\mathbbm{1}_{(-\infty,x)}\cdot p_{T})(0) = \frac{1}{2} - \frac{i}{2}\mathcal H(e^{-i\xi x}\phi_{T}(\xi))(0), \label{eq:HILbert}
\end{align}
where $\mathcal H$ denotes the Hilbert transform
\[
\mathcal Hf(z) = \frac{1}{\pi}\text{PV}\int_\mathbb R \frac{f(y)}{z-y}dy.
\]
Applied to equation \eqref{eq:HILbert}, we have the discrete approximation
\begin{equation}\label{eq:ProbEstSingleBar}
F(x) \approx F_{\Delta_{\xi, N}}(x):=\frac{1}{2} + \frac{i}{2} \sum_{n = -(N-1)}^{N-1} \phi_T((n- \tfrac{1}{2})\Delta_\xi) \frac{\exp(-ix(n-\tfrac{1}{2})\Delta_\xi)}{(n - \tfrac{1}{2})\pi}.
\end{equation}
We then estimate the probability mass of $\ln(S_T/S_0)$ which is neglected by our grid choice, and increase the grid-width by a fixed multiplier $\tau>1$, until a probability tolerance $\epsilon_1 > 0$ is met. Moreover, since a grid for $\ln(S_T/S_0)$  is much larger than required for any individual log return over an increment $\bar \De$, the truncation error for intermediate valuations is negligible. 

For a down-and-out option, if the right tail estimate $|1 -F_{\Delta_\xi,N}(x_{N/2})|\cdot \max\{S_0,W\}> \epsilon_1$, we double the grid size $N$, set $\alpha \gets \tau \alpha$,  and reestimate. For an up-and-out option, we similarly expand the grid as long as $|F_{\Delta_\xi,N}(x_{1})|\cdot \max\{S_0,W\}> \epsilon_1$, which estimates the error in the left tail. Note that in either case, we are only concerned with the probability within the continuation region.

In addition to controlling truncation error we utilize the martingale property of $e^{-(r-q)t}S_{t}$, namely $E[S_{t + \bar \Delta}|S_{t_m}] = S_{t_m}\exp((r-q)\bar \Delta)$, to obtain a proxy for integration error incurred at each step:
\begin{align}\label{eq: EsubNBar}
E_N &:= \Upsilon_{a,N}\cdot \vartheta_*\cdot \sum_{n=1}^{N} \bar \beta_n \exp(x_1 +(n-1)\Delta) 
\end{align}
where $\vartheta_*$ is defined in \eqref{eq: varthetStar}. Here \eqref{eq: EsubNBar}
 approximates the expectation of log-return $\mathbb E[\exp(Y_{\bar \Delta})] = \exp((r-q)\bar \Delta)$ using the projected density (calculated from $\phi_{\bar \Delta}$).  In particular, the error in estimating $\bar \beta_n$, and our ability to accurately calculate intermediate value integrals, is reflected in equation \eqref{eq: EsubNBar}. Hence, once the probability tolerance $\epsilon_1$ is satisfied, we further double the grid size as long as  $|E_N - \exp((r-q)\bar \Delta)|\cdot M > \epsilon_2,$ where the multiplier $M$ accounts for the number of integral approximations made during the algorithm. 

\begin{rem}\label{rem: InitialN}
An initial starting value  $N_x$ is required, which is initialized depending on $\bar \Delta = T/M$. For $\bar \Delta =1/100$ (which includes daily monitoring), we initialize $N_x \gets 2^{10}$. Else if $\bar \Delta\leq 1/40$, we initialize $N_x\gets 2^9$. Else, we initialize $N_x\gets 2^8$.
\end{rem}
 In Algorithm \ref{GridBarrier}, we set $\tau \gets 1.1$, so at each stage the probability threshold is not satisfied, we increase the grid-width by ten percent. We set the tolerance thresholds $\epsilon_1 \gets$5e-08 and $\epsilon_2\gets$1e-05 to conservatively achieve overall valuation error tolerance goal of TOL = 5e-04, which is sufficient for practical purposes. In general, a maximum value of $N$, for example $N=2^{17}$, should be enforced as a stopping criteria for automated parameter selection.

\section{Bounds for the interpolation and truncation errors}\label{err_Marco_Me}
In this section, we list the error bounds and recommendations for the choice of the truncated interval and the discretization step derived in \cite{DeLe14}. Let $\eps>0$ be the error tolerance,
$X$ the L\'evy process with the characteristic exponent $\psi$ admitting analytic continuation to a strip $\{\Im\xi\in (\lm,\lp)\}$. 
\subsection{Truncation error} In the case of the double barrier options, no truncation is needed.
Consider the down-and-out put option with the log-barrier $h=x_1$, strike $K$ and maturity $T$. In \cite{DeLe14}
(Lemma 3.1 and Sect. 3.3.1), it is proved that, at the log-spot $x>h$, we minimize 
\bbe\label{choicexM}
x_M=\frac{-(r+\psi(i\omp))T + \omp' \ln K - \omm x + T(\psi(i\omp)-\psi(i\omm))_+
- \ln(\eps/K-e^h)}{\omp-\omm}
\ee
over $\omm\in (\lm,0)$, $\omp \in (0,\lp)$, e.g., by using an optimization method such as
conjugate gradient or Nelder-Mead. In the case of the down-and-out call option, one
can make the Esscher transform thereby reducing to pricing an option with a payoff uniformly bounded by $K$.
Trivially modifying the proof of Lemma 3.1 in \cite{DeLe14}, one derives the following analog of 
\eq{choicexM}, for down-and-out options  with the payoffs bounded by $K$:
\bbe\label{choicexM2}
x_M=\frac{-(r+\psi(i\omp))T  - \omm x + T(\psi(i\omp)-\psi(i\omm))_+
- \ln(\eps/K)}{\omp-\omm}.
\ee
The case of the up-and-down call options is reducible to the case of the down-and-out put options by the change of measure and symmetry, and the case of up-and-out options with bounded payoffs to the case of down-and-out options by symmetry.

\subsection{Interpolation error (discretization error in the state space)}\label{s:interperror}
Section 3.2 in \cite{DeLe14} contains the bounds and recommendations for options with uniformly bounded payoffs.
To apply the results for call options, one has to  change the measure and reduce to the case of the option with
the bounded payoff.  Lemma 3.2 in \cite{DeLe14} gives an error bound of the quadratic interpolation; 
for the cubic interpolation, see \cite{LeXi12}. 
\begin{lem}\label{l:interpErrMP}
 Let $X$ be a L\'evy process whose transition density $p_{\barDe}$ satisfies $p_{\barDe}\in C^{3}(\bR)$ and
 $p^{(s)}_{\barDe}\in L_1(\bR), s=1,2,3$, then the total error in the time-0 option price due to (piece-wise quadratic) interpolation
 admits the approximate bound via
\[
 \frac{e^{-rT}}{6} \De^{3} (N-1)(K-H) \|p_{\barDe}^{'''}\|_{L_1}.
 \]
 \end{lem}
 To apply Lemma \ref{l:interpErrMP}, one needs the following result (Proposition 3.3 in \cite{DeLe14}).
 \begin{lem}\label{l:discerrDS}
Let $X$ be a model process of order $\nu>0$. Assume that $\nu$ and $\barDe$ are not both small, and the
mesh $\De$ is chosen so that the interpolation error of piece-wise quadratic interpolation is small.
Then the following approximate bound holds
 \begin{equation}\label{e:pl1norm}
\|p^{(n)}_{\barDe}\|_{L_1} \le \rho_n,
\end{equation}
where
\begin{equation}\label{defrhon}
\rho_n=\frac {2 \Ga(n/\nu)}{(d_+^0)^{n/\nu} \pi \nu D(n)} \barDe^{-n/\nu},
\end{equation}
and
\begin{equation}\label{defDn}
D(n)=\sup_{\phi\in\left(0,\min\{\frac{\pi}{2},\frac{\pi}{2\nu}\}\right)} (\cos(\phi\nu))^{n/\nu} \cos(\phi-\pi/2).
\end{equation}
\end{lem}
Applying Lemmas \ref{l:interpErrMP}-\ref{l:discerrDS}, Section 3.2.1 in \cite{DeLe14} gives
the following prescription: 
 \begin{equation}\label{qDe}
\De = \left( \frac{e^{-rT}}{6\epsilon}(N-1)(K-H) \rho_{3}\right)^{-\frac{1}{3}}.
\end{equation}

\section{Auxiliary results}
\subsection{Programming convolution in
MATLAB}\label{sss:convolution-MATLAB}
We borrow this summary from \cite{single}.
\begin{itemize}
\item Inputs: two arrays, $\vec f=(f_j)_{j=1}^M$ and $\vec
g=(g_\ell)_{\ell=1-M}^{M-1}$ (in a program, the indices $\ell$
have to be shifted up).
\item Output: the array $\vec h=(h_k)_{k=1}^M$, whose entries are
given by
%\begin{equation}\label{e:fast-conv}
$h_k = \sum_{j=1}^M f_j g_{k-j}.$
%\end{equation}
\item The calculation:
\begin{eqnarray*}
&& \widetilde{g} = [ \vec g(M:2*M-1) \quad 0 \quad \vec g(1:M-1) ]; \\
&& \widetilde{f} = [ \vec f \quad \mathtt{zeros}(1,M) ]; \\
&& \widetilde{h} = \mathrm{ifft}\left(\mathrm{fft}\widetilde{f}).*\mathrm{fft}(\widetilde{g})\right); \\
&& \vec h = \widetilde{h}(1:M).\end{eqnarray*}
\end{itemize}
For agreement with MATLAB's conventions, we assume that the
array $\vec g$ is also indexed starting with $1$. Thus the first
entry of $\vec g$ is $g_{1-M}$, and the $(2M-1)$-st entry is
$g_{M-1}$.

MATLAB has a built-in \texttt{conv} function, which computes
a certain version of the discrete convolution of two arrays.
However, the realization of this function in MATLAB is also based on
FFT, and typically, it is more convenient to use the algorithm
presented above instead of \texttt{conv}.
% 
%\footnotesize
%\bibliography{BSINHbib}

\begin{thebibliography}{}

\bibitem[Abate and Valko, 2004]{AbateValko04}
Abate, J. and Valko, P. (2004).
\newblock Multi-precision {L}aplace inversion.
\newblock {\em International Journal of Numerical Methods in Engineering},
  60:979--993.

\bibitem[Abate and Whitt, 1992a]{AbWh}
Abate, J. and Whitt, W. (1992a).
\newblock The {F}ourier-series method for inverting transforms of probability
  distributions.
\newblock {\em Queueing Systems}, 10:5--88.

\bibitem[Abate and Whitt, 1992b]{AbWh92OR}
Abate, J. and Whitt, W. (1992b).
\newblock Numerical inversion of of probability generating functions.
\newblock {\em Operation Research Letters}, 12:245--251.

\bibitem[Andricopoulos et~al., 2003]{QUAD03}
Andricopoulos, A., Widdicks, D., Newton, D., and Duck, P. (2003).
\newblock Universal option valuation using quadrature methods.
\newblock {\em J. Finan. Econ.}, 67:447--471.

\bibitem[Barndorff-Nielsen, 1998]{BN98}
Barndorff-Nielsen, O. (1998).
\newblock Processes of normal inverse {G}aussian type.
\newblock {\em Finance and Stochastics}, 2(1):41--68.

\bibitem[Barndorff-Nielsen and Levendorski\v{i}, 2001]{B-N-L}
Barndorff-Nielsen, O. and Levendorski\v{i}, S. (2001).
\newblock Feller {P}rocesses of {N}ormal {I}nverse {G}aussian type.
\newblock {\em Quantitative Finance}, 1:318--331.

\bibitem[Boyarchenko et~al., 2011]{BIL}
Boyarchenko, M., de~Innocentis, M., and Levendorski\u{i}, S. (2011).
\newblock Prices of barrier and first-touch digital options in {L}\'evy-driven
  models, near barrier.
\newblock {\em International Journal of Theoretical and Applied Finance},
  14(7):1045--1090.
\newblock Available at SSRN: http://papers.ssrn.com/abstract=1514025.

\bibitem[Boyarchenko and Levendorski\u{i}, 2009]{single}
Boyarchenko, M. and Levendorski\u{i}, S. (2009).
\newblock Prices and sensitivities of barrier and first-touch digital options
  in {L}\'evy-driven models.
\newblock {\em International Journal of Theoretical and Applied Finance},
  12(8):1125--1170.

\bibitem[Boyarchenko and Levendorski\u{i}, 2012]{BLdouble}
Boyarchenko, M. and Levendorski\u{i}, S. (2012).
\newblock Valuation of continuously monitored double barrier options and
  related securities.
\newblock {\em Mathematical Finance}, 22(3):419--444.

\bibitem[Boyarchenko and Levendorski\u{i}, 2015]{one-sidedCDS}
Boyarchenko, M. and Levendorski\u{i}, S. (2015).
\newblock Ghost {C}alibration and {P}ricing {B}arrier {O}ptions and {C}redit
  {D}efault {S}waps in spectrally one-sided {L}\'evy models: The {P}arabolic
  {L}aplace {I}nversion {M}ethod.
\newblock {\em Quantitative Finance}, 15(3):421--441.
\newblock Available at SSRN: http://ssrn.com/abstract=2445318.

\bibitem[Boyarchenko and Levendorski\u{i}, 1998]{BL-FT}
Boyarchenko, S. and Levendorski\u{i}, S. (1998).
\newblock On rational pricing of derivative securities for a family of
  non-{G}aussian processes.
\newblock Preprint 98/7, Institut f{\"u}r Mathematik, Universit{\"a}t Potsdam.
\newblock Available at http://opus.kobv.de/ubp/volltexte/2008/2519/.

\bibitem[Boyarchenko and Levendorski\u{i}, 1999]{genBS}
Boyarchenko, S. and Levendorski\u{i}, S. (1999).
\newblock Generalizations of the {B}lack-{S}choles equation for truncated
  {L}\'evy processes.
\newblock Working Paper, University of Pennsylvania.

\bibitem[Boyarchenko and Levendorski\u{i}, 2000]{KoBoL}
Boyarchenko, S. and Levendorski\u{i}, S. (2000).
\newblock Option pricing for truncated {L}\'evy processes.
\newblock {\em International Journal of Theoretical and Applied Finance},
  3(3):549--552.

\bibitem[Boyarchenko and Levendorski\u{i}, 2002]{NG-MBS}
Boyarchenko, S. and Levendorski\u{i}, S. (2002).
\newblock {\em Non-{G}aussian {M}erton-{B}lack-{S}choles {T}heory}, volume~9 of
  {\em Adv. Ser. Stat. Sci. Appl. Probab.}
\newblock World Scientific Publishing Co., River Edge, NJ.

\bibitem[Boyarchenko and Levendorski\u{i}, 2011]{iFT0}
Boyarchenko, S. and Levendorski\u{i}, S. (2011).
\newblock New efficient versions of fourier transform method in applications to
  option pricing.
\newblock Working paper.
\newblock Available at SSRN: http://ssrn.com/abstract=1846633.

\bibitem[Boyarchenko and Levendorski\u{i}, 2013]{paraLaplace}
Boyarchenko, S. and Levendorski\u{i}, S. (2013).
\newblock Efficient {L}aplace inversion, {W}iener-{H}opf factorization and
  pricing lookbacks.
\newblock {\em International Journal of Theoretical and Applied Finance},
  16(3):1350011 (40 pages).
\newblock Available at SSRN: http://ssrn.com/abstract=1979227.

\bibitem[Boyarchenko and Levendorski\u{i}, 2014]{iFT}
Boyarchenko, S. and Levendorski\u{i}, S. (2014).
\newblock Efficient variations of {F}ourier transform in applications to option
  pricing.
\newblock {\em J. Computational Finance}, 18(2):57--90.

\bibitem[Boyarchenko and Levendorski\u{i}, 2017]{BarrStIR}
Boyarchenko, S. and Levendorski\u{i}, S. (2017).
\newblock Efficient pricing barrier options and {C}{D}{S} in {L}\'evy models
  with stochastic interest rate.
\newblock {\em Mathematical Finance}, 27(4):1089--1123.
\newblock DOI: 10.1111/mafi.12121.

\bibitem[Boyarchenko and Levendorski\u{i}, 2019]{BoyarLevenSinh19}
Boyarchenko, S. and Levendorski\u{i}, S. (2019).
\newblock {SINH}-acceleration: Efficient evaluation of probability
  distributions, option pricing, and {M}onte {C}arlo simulations.
\newblock {\em Int. J. Theoretical and Applied Finance}, 22(2):49 pages.

\bibitem[Boyarchenko and Levendorski\u{i}, 2020a]{ConfAccelerationStable}
Boyarchenko, S. and Levendorski\u{i}, S. (2020a).
\newblock Conformal accelerations method and efficient evaluation of stable
  distributions.
\newblock {\em Acta Applicandae Mathematicae}, 169:711--765.
\newblock Available at SSRN: https://ssrn.com/abstract=3206696 or
  http://dx.doi.org/10.2139/ssrn.3206696.

\bibitem[Boyarchenko and Levendorski\u{i}, 2020b]{Contrarian}
Boyarchenko, S. and Levendorski\u{i}, S. (2020b).
\newblock Static and semi-static hedging as contrarian or conformist bets.
\newblock {\em Mathematical Finance}, 3(30):921--960.
\newblock Available at SSRN: https://ssrn.com/abstract=3329694 or
  http://arxiv.org/abs/1902.02854.

\bibitem[Briggs and Henson, 2005]{BriggsHenson}
Briggs, W. and Henson, V. (2005).
\newblock {\em The {DFT}: {A}n {O}wners' {M}anual for the {D}iscrete {F}ourier
  {T}ransform}.
\newblock SIAM, Philadelphia.

\bibitem[Carr et~al., 2002]{CGMY}
Carr, P., Geman, H., Madan, D.~B., and Yor, M. (2002).
\newblock The fine structure of asset returns: an empirical investigation.
\newblock {\em J. Business}, 75:305--332.

\bibitem[Carr and Madan, 1999]{carr-madan-FFT}
Carr, P. and Madan, D. (1999).
\newblock Option valuation using the {F}ast {F}ourier {T}ransform.
\newblock {\em Journal of Computational Finance}, 2(4):61--73.

\bibitem[Carr and Madan, 2009]{Carr-Madan-saddlepoint}
Carr, P. and Madan, D. (2009).
\newblock Saddlepoint methods for option pricing.
\newblock {\em Journal of Computational Finance}, 13(1):49--61.

\bibitem[Christensen, 2003]{OC03}
Christensen, O. (2003).
\newblock {\em An Introduction to Frames and {R}iesz Bases}.
\newblock Birkhauser Boston.

\bibitem[Cui et~al., 2017a]{CuiKirkNguyenCliquet17}
Cui, Z., Kirkby, J., and Nguyen, D. (2017a).
\newblock Equity-linked annuity pricing with cliquet-style guarantees in
  regime-switching and stochastic volatility models with jumps.
\newblock {\em Insurance: Mathematics and Economics}, 74:46--62.

\bibitem[Cui et~al., 2017b]{CuiKirkNguyenDVSwap17}
Cui, Z., Kirkby, J., and Nguyen, D. (2017b).
\newblock A general framework for discretely sampled realized variance
  derivatives in stochastic volatility models with jumps.
\newblock {\em European J. Operational Research}, 262(1):381--400.

\bibitem[Cui et~al., 2021]{cui2021data}
Cui, Z., Kirkby, J.~L., and Nguyen, D. (2021).
\newblock A data-driven framework for consistent financial valuation and risk
  measurement.
\newblock {\em European Journal of Operational Research}, 289(1):381--398.

\bibitem[Duffie et~al., 2000]{DPS}
Duffie, D., Pan, J., and Singleton, K. (2000).
\newblock Transform {A}nalysis and {A}sset {P}ricing for {A}ffine {J}ump
  {D}iffusions.
\newblock {\em Econometrica}, 68:1343--1376.

\bibitem[Duistermaat, 1995]{Duistermaat95}
Duistermaat, J. (1995).
\newblock {\em Fourier Integral Operators}.
\newblock Progress in Mathematics. Birkh{\"a}user, Basel.

\bibitem[Eskin, 1981]{eskin}
Eskin, G. (1981).
\newblock {\em Boundary {V}alue {P}roblems for {E}lliptic {P}seudodifferential
  {E}quations}, volume~9 of {\em Transl. Math. Monogr.}
\newblock American Mathematical Society, Providence, RI.

\bibitem[Eydeland, 1994]{eydeland}
Eydeland, A. (1994).
\newblock A fast algorithm for computing integrals in function spaces:
  financial applications.
\newblock {\em Computational Economics}, 7:277--285.

\bibitem[Eydeland and Mahoney, 2001]{eydeland-mahoney}
Eydeland, A. and Mahoney, D. (2001).
\newblock The grid model for derivative pricing.
\newblock {\em Mirant Technical Report}.

\bibitem[Fang and Oosterlee, 2008]{FaOo08}
Fang, F. and Oosterlee, C. (2008).
\newblock A novel pricing method for {E}uropean options based on {F}ourier
  cosine series expansions.
\newblock {\em SIAM J. Sci. Comput.}, 31:826--848.

\bibitem[Fedoryuk, 1987]{Fedoryuk}
Fedoryuk, M. (1987).
\newblock {\em Asymptotic: Integrals and Series}.
\newblock Nauka, Moscow.
\newblock In Russian.

\bibitem[Feng and Linetsky, 2008]{FeLi08}
Feng, L. and Linetsky, V. (July 2008).
\newblock Pricing discretely monitored barrier options and defaultable bonds in
  {L}evy process models: a fast {H}ilbert transform approach.
\newblock {\em Math. Finan.}, 18(3):337–384.

\bibitem[Fusai et~al., 2006]{FusaiBarr}
Fusai, G., Abrahams, I., and Sgarra, C. (2006).
\newblock An exact analytical solution for discrete barrier options.
\newblock {\em Finance and Stochastics}, 10(1):1--26.

\bibitem[Fusai et~al., 2016]{FusaiGermanoMarazzina}
Fusai, G., Germano, G., and Marazzina, D. (2016).
\newblock Spitzer identity, {W}iener-{H}opf factorization and pricing of
  discretely monitored exotic options.
\newblock {\em European Journal of Operational Research}, 251(1):124--134.
\newblock DOI:10.1016/j.ejor.2015.11.027.

\bibitem[Gasquet and Witomski, 1998]{GasWit98}
Gasquet, C. and Witomski, P. (1998).
\newblock {\em Fourier Analysis and applications: filtering, numerical
  computation, wavelets}.
\newblock Springer, New York.

\bibitem[Heil, 2011]{CH11}
Heil, C. (2011).
\newblock {\em A Basis Theory Primer}.
\newblock Birkhauser, expanded edition.

\bibitem[Heston, 1993]{Heston93}
Heston, S. (1993).
\newblock A closed-form solution for options with stochastic volatility with
  applications to bond and currency options.
\newblock {\em Rev. Financ. Studies}, 6:327--343.

\bibitem[H{\"o}rmander, 1985]{Hormander}
H{\"o}rmander, L. (1985).
\newblock {\em The {A}nalysis of Linear {P}artial {D}ifferential {O}perators
  III: {P}seudo-{D}ifferential {O}perators}.
\newblock Springer, Berlin.

\bibitem[Innocentis and Levendorski\u{i}, 2014]{DeLe14}
Innocentis, M. and Levendorski\u{i}, S. (2014).
\newblock Pricing discrete barrier options and credit default swaps under
  {L}\'evy processes.
\newblock {\em Quantitative Finance}, 14(8):1337--1365.
\newblock Available at: DOI:10.1080/14697688.2013.826814.

\bibitem[Innocentis and Levendorski\u{i}, 2016]{HestonCalibMarcoMe}
Innocentis, M. and Levendorski\u{i}, S. (2016).
\newblock Calibration and {B}acktesting of the {H}eston {M}odel for
  {C}ounterparty {C}redit {R}isk.
\newblock Working paper.
\newblock Available at SSRN: http://ssrn.com/abstract=2757008.

\bibitem[Innocentis and Levendorski\u{i}, 2017]{HestonCalibMarcoMeRisk}
Innocentis, M. and Levendorski\u{i}, S. (2017).
\newblock Calibration {H}eston {M}odel for {C}redit {R}isk.
\newblock {\em Risk}, pages 90--95.

\bibitem[Kirkby, 2015]{Ki14}
Kirkby, J. (2015).
\newblock Efficient option pricing by frame duality with the fast {F}ourier
  transform.
\newblock {\em SIAM J. Financial Mathematics}, 6(1):713--747.

\bibitem[Kirkby, 2016]{Ki14B}
Kirkby, J. (2016).
\newblock An efficient transform method for {A}sian option pricing.
\newblock {\em SIAM J. Financial Mathematics}, 7(1):845--892.

\bibitem[Kirkby, 2017a]{Ki14C}
Kirkby, J. (2017a).
\newblock Robust barrier option pricing by frame projection under exponential
  {L}\'evy dynamics.
\newblock {\em Applied Mathematical Finance}, 24(4):337--386.

\bibitem[Kirkby, 2017b]{Ki14E}
Kirkby, J. (2017b).
\newblock Robust option pricing with characteristic functions and the
  {B}-spline order of density projection.
\newblock {\em J. Computational Finance}, 21(2):101--127.

\bibitem[Kirkby, 2018]{Ki16A}
Kirkby, J. (2018).
\newblock American and exotic option pricing with jump diffusions and other
  {L}evy processes.
\newblock {\em J. Computational Finance}, 22(3):89--148.

\bibitem[Kirkby and Deng, 2019a]{KiDe14}
Kirkby, J. and Deng, S. (2019a).
\newblock Static hedging and pricing of exotic options with payoff frames.
\newblock {\em Mathematical Finance}, 29(2):612--658.

\bibitem[Kirkby and Deng, 2019b]{lars2019swing}
Kirkby, J. and Deng, S.-J. (2019b).
\newblock Swing option pricing by dynamic programming with {B}-spline density
  projection.
\newblock {\em International Journal of Theoretical and Applied Finance},
  22(08):1950038.

\bibitem[Kirkby et~al., 2017]{KirkNguyenCuiBermuBar17}
Kirkby, J., Nguyen, D., and Cui, Z. (2017).
\newblock A unified approach to {B}ermudan and barrier options under stochastic
  volatility models with jumps.
\newblock {\em J. Economic Dynamics and Control}, 80:75--100.

\bibitem[Kirkby et~al., 2020]{kirkby2020analysis}
Kirkby, J.~L., Mitra, S., and Nguyen, D. (2020).
\newblock An analysis of dollar cost averaging and market timing investment
  strategies.
\newblock {\em European Journal of Operational Research}, 286(3):1168--1186.

\bibitem[Kirkby and Nguyen, 2021]{kirkby2021equity}
Kirkby, J.~L. and Nguyen, D. (2021).
\newblock Equity-linked guaranteed minimum death benefits with dollar cost
  averaging.
\newblock {\em Insurance: Mathematics and Economics}, 100:408--428.

\bibitem[Kuznetsov, 2010]{KuznetsovWien10}
Kuznetsov, A. (2010).
\newblock Wiener-{H}opf factorization and distribution of extrema for a family
  of l\'evy processes.
\newblock {\em Annals of Applied Probability}, 20:1801--1830.

\bibitem[Lee, 2004]{Lee04}
Lee, R. (2004).
\newblock Option pricing by transform methods: extensions, unification, and
  error control.
\newblock {\em J. Comput. Finance}, 7(3):50--86.

\bibitem[Leentvaar and Oosterlee, 2008]{CONV2}
Leentvaar, C. and Oosterlee, C. (2008).
\newblock Multi-asset option pricing using a parallel {F}ourier-based
  technique.
\newblock {\em J. Comput. Fin.}, 12(1):1--26.

\bibitem[Levendorski\u{i}, 1993]{DegEllEq}
Levendorski\u{i}, S. (1993).
\newblock {\em Degenerate {E}lliptic {E}quations}, volume 258 of {\em
  Mathematics and its {A}pplications}.
\newblock Kluwer {A}cademic {P}ublishers {G}roup, Dordrecht.

\bibitem[Levendorski\u{i}, 2004a]{early-exercise}
Levendorski\u{i}, S. (2004a).
\newblock Early exercise boundary and option pricing in {L}\'evy driven models.
\newblock {\em Quantitative Finance}, 4(5):525--547.

\bibitem[Levendorski\u{i}, 2004b]{amer-put-levy}
Levendorski\u{i}, S. (2004b).
\newblock Pricing of the {A}merican put under {L}\'evy processes.
\newblock {\em International Journal of Theoretical and Applied Finance},
  7(3):303--335.

\bibitem[Levendorski\u{i}, 2012]{paraHeston}
Levendorski\u{i}, S. (2012).
\newblock Efficient pricing and reliable calibration in the {H}eston model.
\newblock {\em International Journal of Theoretical and Applied Finance},
  15(7).
\newblock 125050 (44 pages).

\bibitem[Levendorski\u{i}, 2014]{paired}
Levendorski\u{i}, S. (2014).
\newblock Method of paired contours and pricing barrier options and {C}{D}{S}
  of long maturities.
\newblock {\em International Journal of Theoretical and Applied Finance},
  17(5):1--58.
\newblock 1450033 (58 pages).

\bibitem[Levendorski\u{i}, 2016a]{Sinh}
Levendorski\u{i}, S. (2016a).
\newblock {F}ractional-{P}arabolic {D}eformations with {S}inh-{A}cceleration.
\newblock Working paper.
\newblock Available at SSRN: http://ssrn.com/abstract=2758811.

\bibitem[Levendorski\u{i}, 2016b]{pitfalls}
Levendorski\u{i}, S. (2016b).
\newblock Pitfalls of the {F}ourier {T}ransform method in {A}ffine {M}odels,
  and remedies.
\newblock {\em Applied Mathematical Finance}, 23(2):81--134.
\newblock Available at http://dx.doi.org/10.1080/1350486X.2016.1159918 or
  http://ssrn.com/abstract=2367547.

\bibitem[Levendorski\u{i}, 2018]{AsianGammaSIAMFM}
Levendorski\u{i}, S. (2018).
\newblock Pricing arithmetic {A}sian options under {L}\'evy models by backward
  induction in the dual space.
\newblock {\em SIAM FM}, 9(1):1--27.

\bibitem[Levendorski\u{i} and Xie, 2012a]{IAC}
Levendorski\u{i}, S. and Xie, J. (2012a).
\newblock Fast pricing and calculation of sensitivities of {O}{T}{M} {E}uropean
  options under {L}\'evy processes.
\newblock {\em Journal of Computational Finance}, 15(2).
\newblock Available at SSRN:http://ssrn.com/abstract=1589809.

\bibitem[Levendorski\u{i} and Xie, 2012b]{LeXi12}
Levendorski\u{i}, S. and Xie, J. (2012b).
\newblock Pricing discretely sampled {A}sian options under {L}evy processes.
\newblock Available at SSRN: http://papers.ssrn.com/abstract=2088214.

\bibitem[Lord et~al., 2008]{CONV}
Lord, R., Fang, F., Bervoets, F., and Oosterlee, C. (2008).
\newblock A fast and accurate {FFT}-based method for pricing early-exercise
  options under {L}evy processes.
\newblock {\em SIAM J. Sci. Comput.}, 10:1678--1705.

\bibitem[Madan and Seneta, 1990]{MaSe90}
Madan, D. and Seneta, E. (1990).
\newblock The variance gamma (v.g.) model for share market returns.
\newblock {\em J. Business}, 63:511--524.

\bibitem[Madan et~al., 1998]{MGC98}
Madan, D.~B., Carr, P., and Chang, E. (1998).
\newblock The variance gamma process and option pricing.
\newblock {\em European Finance Review}, 2:79--105.

\bibitem[Ruijter et~al., 2015]{ruijter2015application}
Ruijter, M., Versteegh, M., and Oosterlee, C.~W. (2015).
\newblock On the application of spectral filters in a fourier option pricing
  technique.
\newblock {\em Journal of Computational Finance}, 19(1):75--106.

\bibitem[Sato, 1999]{Sa99}
Sato, K.-I. (1999).
\newblock {\em Levy Processes and Infinitely Divisible Distributions}.
\newblock Cambridge University Press, Cambridge, UK.

\bibitem[Shi and Zhang, 2021]{ZhangShi2020}
Shi, B. and Zhang, Z. (2021).
\newblock Pricing {EIA} with cliquet-style guarantees under time-changed {L}evy
  models by frame duality projection.
\newblock {\em Comm. in Nonlinear Science and Numerical Sim.}, 95:105651.

\bibitem[Stenger, 1993a]{stenger-book}
Stenger, F. (1993a).
\newblock {\em Numerical {M}ethods based on {S}inc and {A}nalytic functions}.
\newblock Springer-Verlag, New York.

\bibitem[Stenger, 1993b]{St93}
Stenger, F. (1993b).
\newblock {\em Numerical Methods based on Sinc and Analytic functions}.
\newblock Springer-Verlag, New York.

\bibitem[Stenger, 2000]{stengerreview00}
Stenger, F. (2000).
\newblock Summary of {S}inc numerical methods.
\newblock {\em Journal of Computational and Applied Mathematics}, 121:379--420.

\bibitem[Unser and Daubechies, 1997]{Un97}
Unser, M. and Daubechies, I. (July 1997).
\newblock On the approximation power of convolution-based least squares versus
  interpolation.
\newblock {\em {IEEE} Transactions on Signal Processing}, 45(7):1697--1711.

\bibitem[Wang and Zhang, 2019]{WangZhang19}
Wang, W. and Zhang, Z. (2019).
\newblock Computing the {G}erber-{S}hiu function by frame duality projection.
\newblock {\em Scandanavian Actuarial Journal}, 2019(4):291--307.

\bibitem[Young, 1980]{YO80}
Young, R. (1980).
\newblock {\em An Introduction to Nonharmonic {F}ourier Series}.
\newblock Academic Press, New York, (revised) edition.

\bibitem[Zhang et~al., 2020]{ZhangYonYu2020}
Zhang, Z., Yong, Y., and Yu, W. (2020).
\newblock Valuing equity-linked death benefits in general exponential {L}\'evy
  models.
\newblock {\em J. Computational and Applied Math.}, 365:112377.

\end{thebibliography}
%\bibliographystyle{plain}
 %\bibliographystyle{apalike} 

\end{document}